\documentclass[12pt]{article}
\usepackage[margin=1in]{geometry}
\usepackage[utf8]{inputenc}
\usepackage[english]{babel}
\usepackage[dvipsnames]{xcolor}
\usepackage[normalem]{ulem}
\usepackage[labelfont=bf]{caption}
\usepackage{algorithm, algorithmic}
\usepackage[numbers]{natbib}
\usepackage{microtype, fancyhdr, siunitx, graphicx, hyperref, epstopdf, tikz, adjustbox, mathtools, float, anyfontsize}
\usepackage{bm, amssymb, amsmath, amsthm}
\usepackage{etoolbox}
\usepackage{multirow}
\newcommand{\I}{\bm{I}}   % Identity matrix
\newcommand{\E}{\mathbb{E}}    % Expectation
\newcommand{\V}{\mathbb{V}}    % Variance
\newcommand{\R}{\mathbb{R}}    % Real numbers
\newcommand{\C}{\mathbb{C}}    % Complex numbers
\newcommand{\bO}{\mathcal{O}}  % "big O"
\newcommand{\Fr}{\mathcal{F}}  % Fourier Transform operator
\newcommand{\Nd}{\mathcal{N}}  % Normal distribution
\newcommand{\eps}{\varepsilon} % epsilon
\newcommand{\lam}{\lambda}     % lambda
\newcommand{\sgm}{\sigma}      % sigma
\newcommand{\omg}{\omega}      % omega
\newcommand{\bx}{\bm{x}}       % bold x
\newcommand{\by}{\bm{y}}       % bold y
\newcommand{\bv}{\bm{v}}       % bold v
\newcommand{\bu}{\bm{u}}       % bold u
\newcommand{\bomg}{\bm{\omega}} % bold mu
\newcommand{\bS}{\bm{\Sigma}}  % bold Sigma
\newcommand{\abs}[1]{\left|#1\right|}      % absolute value
\newcommand{\cvd}{\rightsquigarrow}        % convergence in distribution
\newcommand{\ind}[1]{\bm{1}_{\set{#1}}}    % Indicator function
\newcommand{\dif}{\mathop{}\!\mathrm{d}}   % d for dx in integral
\newcommand{\norm}[2][]{\left\Vert#2\right\Vert_{#1}} % norm
\newcommand{\shortnorm}[2][]{\Vert#2\Vert_{#1}}       % short norm
\newcommand{\set}[1]{\left\{ #1\right\}}              % set notation
\newcommand{\inp}[1]{\left\langle #1 \right\rangle}   % inner product
\newcommand{\mat}[1]{\begin{bmatrix}#1\end{bmatrix}}  % matrix

\renewcommand{\epsilon}{\varepsilon}
\renewcommand{\phi}{\varphi}

\DeclareMathOperator*{\spn}{span}
\DeclareMathOperator*{\sinc}{sinc}
 
\DeclareMathOperator*{\supp}{supp}     % logdet
  
\newcommand{\iid}{\mathrel{\overset{\scalebox{0.5}{\text{i.i.d.}}}{\scalebox{1.1}[1]{$\sim$}}}}
\newtheorem{theorem}{Theorem}
\newtheorem*{theorem*}{Theorem}
\newtheorem{lemma}{Lemma}
\newtheorem*{lemma*}{Lemma}
\newtheorem{corollary}{Corollary}
\newtheorem*{corollary*}{Corollary}

\newtheorem*{definition*}{Definition}

\newtheorem*{proposition*}{Proposition}

\newtheorem*{conjecture*}{Conjecture}
\theoremstyle{remark}

\newtheorem*{remark*}{Remark}

\numberwithin{equation}{section}

\newcommand{\Frt}{\bm{\mathcal{F}}}  % Fourier Transform matrix

\newcommand{\bal}{\bm{\alpha}}
\newcommand{\bbeta}{\bm{\beta}}

\graphicspath{{./figures/}}

\newbool{anon}
\setbool{anon}{false}

\newbool{figs}
\setbool{figs}{true}

\def\spacingset#1{\renewcommand{\baselinestretch}%
{#1}\small\normalsize} \spacingset{1}

\title{
  \textbf{Fast nonparametric spectral density estimation from irregularly sampled data}
} \ifbool{anon}{}{ 
\author{ 
Christopher J. Geoga\thanks{corresponding author: \texttt{geoga@wisc.edu}} \\
\small Department of Statistics \\ 
\small University of Wisconsin-Madison \\
\small Madison, WI 53706
\and 
Paul G. Beckman  
\thanks{PGB is supported in part by the
Office of Naval Research under award \#N00014-21-1-2383 and by the U.S.
Department of Energy, Office of Science, Office of Advanced Scientific Computing
Research, Department of Energy Computational Science Graduate Fellowship under
Award Number DE-SC0022158} \\ 
\small Courant Institute \\ 
\small New York University \\ 
\small New York, NY 10012 \\
}
}
\date{}

\begin{document}

\maketitle

\begin{abstract}
We introduce a nonparametric spectral density estimator for continuous-time and
continuous-space processes measured at fully irregular locations. Our estimator
is constructed using a weighted nonuniform Fourier sum whose weights yield a
high-accuracy quadrature rule with respect to a user-specified window function.
The resulting estimator significantly reduces the aliasing seen in periodogram
approaches and least squares spectral analysis, sidesteps the dangers of
ill-conditioning of the nonuniform Fourier inverse problem, and can be adapted
to a wide variety of irregular sampling settings. We describe methods for
rapidly computing the necessary weights in various settings, making the
estimator scalable to large datasets. We then provide a theoretical
analysis of sources of bias, and close with demonstrations of the method's
efficacy, including for processes that exhibit very slow spectral decay and are
observed at up to a million locations in multiple dimensions.
  
  \bigskip

\end{abstract}

\newpage

\section{Introduction} \label{sec:intro}

Let $Y(\bx)$ be a stochastic process indexed by $\bx \in \R^d$ with finite
second moments. Consider the case of a \emph{weakly stationary} $Y(\bx)$ with
$\E Y(\bx) \equiv 0$, so that its covariance function $K(\bx, \bx') :=
\text{Cov}(Y(\bx), Y(\bx'))$ is shift-invariant, meaning $K(\bx, \bx') = K(\bx +
\bm{h}, \bx' + \bm{h})$ for all $\bm{h} \in \R^d$.  In this setting, we may
study the covariance function in a single-argument form as $K(\bx - \bx') =
\text{Cov}(Y(\bx), Y(\bx'))$. By the Wiener-Khinchin theorem
\cite{khintchine1934}, if $K$ is real-valued and continuous, then there exists a
function $S$ that is non-negative, integrable, and symmetric about the origin
(so that $S(\bomg) = S(-\bomg)$) such that
\begin{equation} \label{eq:sdf}
  K(\bx - \bx') = \int_{\R^d} e^{2 \pi i \bomg^T (\bx - \bx')} S(\bomg) \dif
\bomg.
\end{equation}
This function $S$ is called the \emph{spectral density}, and features of $S$
such as integrable singularities or finite moments $\int_{\R^d}
\norm[2]{\bomg}^{2 s} S(\bomg) \dif \bomg < \infty$ have precise implications
regarding sample-path properties of $Y(\bx)$. For these reasons, a large amount
of the theory for prediction, interpolation, and estimation in time series
analysis and Gaussian processes is expressed in the the spectral domain.
See~\cite{brockwell1991, stein1999} for further discussion.

Considering the interpretive value of $S$ and the computational efficiency of
the fast Fourier transform (FFT), the task of obtaining nonparametric estimators
for $S$ is a large area of research. Tools for this problem are particularly
mature for gridded data in one dimension, where a wide variety of methods for
specific applications have been developed and provide excellent results (see
\cite{thomson1982,riedel1995,prieto2007,orfanidis1995,percival2020} and
references therein). In the case of irregular sampling where $\set{x_j}_{j=1}^n$
is not a subset of some Cartesian grid, however, the spectral estimation problem
becomes much more challenging. The primary complication is that one typically
treats such measurements as coming from a continuous-time (or continuous
space-time) process. Unlike the case of truly discrete time series at regular
sampling intervals, in which the spectral density can be represented on the
finite interval $[-\frac{1}{2 \Delta t}, \frac{1}{2 \Delta t}]$, for
continuous-time models $S$ is supported on $\R$ (or $\R^d$ for multi-dimensional
processes). By information-theoretic limits like the Shannon-Nyquist sampling
theorem \cite{shannon1949, landau1967sampling}, the spectral density $S$ may be
supported on frequency bands that are unresolvable from observations at the
given locations $\{x_j\}_{j=1}^n$. 

\begin{figure}[!t]
  \centering
% GNUPLOT: LaTeX picture with Postscript
\begingroup
  \makeatletter
  \providecommand\color[2][]{%
    \GenericError{(gnuplot) \space\space\space\@spaces}{%
      Package color not loaded in conjunction with
      terminal option `colourtext'%
    }{See the gnuplot documentation for explanation.%
    }{Either use 'blacktext' in gnuplot or load the package
      color.sty in LaTeX.}%
    \renewcommand\color[2][]{}%
  }%
  \providecommand\includegraphics[2][]{%
    \GenericError{(gnuplot) \space\space\space\@spaces}{%
      Package graphicx or graphics not loaded%
    }{See the gnuplot documentation for explanation.%
    }{The gnuplot epslatex terminal needs graphicx.sty or graphics.sty.}%
    \renewcommand\includegraphics[2][]{}%
  }%
  \providecommand\rotatebox[2]{#2}%
  \@ifundefined{ifGPcolor}{%
    \newif\ifGPcolor
    \GPcolortrue
  }{}%
  \@ifundefined{ifGPblacktext}{%
    \newif\ifGPblacktext
    \GPblacktexttrue
  }{}%
  % define a \g@addto@macro without @ in the name:
  \let\gplgaddtomacro\g@addto@macro
  % define empty templates for all commands taking text:
  \gdef\gplbacktext{}%
  \gdef\gplfronttext{}%
  \makeatother
  \ifGPblacktext
    % no textcolor at all
    \def\colorrgb#1{}%
    \def\colorgray#1{}%
  \else
    % gray or color?
    \ifGPcolor
      \def\colorrgb#1{\color[rgb]{#1}}%
      \def\colorgray#1{\color[gray]{#1}}%
      \expandafter\def\csname LTw\endcsname{\color{white}}%
      \expandafter\def\csname LTb\endcsname{\color{black}}%
      \expandafter\def\csname LTa\endcsname{\color{black}}%
      \expandafter\def\csname LT0\endcsname{\color[rgb]{1,0,0}}%
      \expandafter\def\csname LT1\endcsname{\color[rgb]{0,1,0}}%
      \expandafter\def\csname LT2\endcsname{\color[rgb]{0,0,1}}%
      \expandafter\def\csname LT3\endcsname{\color[rgb]{1,0,1}}%
      \expandafter\def\csname LT4\endcsname{\color[rgb]{0,1,1}}%
      \expandafter\def\csname LT5\endcsname{\color[rgb]{1,1,0}}%
      \expandafter\def\csname LT6\endcsname{\color[rgb]{0,0,0}}%
      \expandafter\def\csname LT7\endcsname{\color[rgb]{1,0.3,0}}%
      \expandafter\def\csname LT8\endcsname{\color[rgb]{0.5,0.5,0.5}}%
    \else
      % gray
      \def\colorrgb#1{\color{black}}%
      \def\colorgray#1{\color[gray]{#1}}%
      \expandafter\def\csname LTw\endcsname{\color{white}}%
      \expandafter\def\csname LTb\endcsname{\color{black}}%
      \expandafter\def\csname LTa\endcsname{\color{black}}%
      \expandafter\def\csname LT0\endcsname{\color{black}}%
      \expandafter\def\csname LT1\endcsname{\color{black}}%
      \expandafter\def\csname LT2\endcsname{\color{black}}%
      \expandafter\def\csname LT3\endcsname{\color{black}}%
      \expandafter\def\csname LT4\endcsname{\color{black}}%
      \expandafter\def\csname LT5\endcsname{\color{black}}%
      \expandafter\def\csname LT6\endcsname{\color{black}}%
      \expandafter\def\csname LT7\endcsname{\color{black}}%
      \expandafter\def\csname LT8\endcsname{\color{black}}%
    \fi
  \fi
    \setlength{\unitlength}{0.0500bp}%
    \ifx\gptboxheight\undefined%
      \newlength{\gptboxheight}%
      \newlength{\gptboxwidth}%
      \newsavebox{\gptboxtext}%
    \fi%
    \setlength{\fboxrule}{0.5pt}%
    \setlength{\fboxsep}{1pt}%
    \definecolor{tbcol}{rgb}{1,1,1}%
\begin{picture}(9060.00,3400.00)%
    \gplgaddtomacro\gplbacktext{%
      \csname LTb\endcsname%%
      \put(803,676){\makebox(0,0)[r]{\strut{}\footnotesize $10^{-10}$}}%
      \csname LTb\endcsname%%
      \put(803,920){\makebox(0,0)[r]{\strut{}\footnotesize $10^{-9}$}}%
      \csname LTb\endcsname%%
      \put(803,1164){\makebox(0,0)[r]{\strut{}\footnotesize $10^{-8}$}}%
      \csname LTb\endcsname%%
      \put(803,1408){\makebox(0,0)[r]{\strut{}\footnotesize $10^{-7}$}}%
      \csname LTb\endcsname%%
      \put(803,1652){\makebox(0,0)[r]{\strut{}\footnotesize $10^{-6}$}}%
      \csname LTb\endcsname%%
      \put(803,1896){\makebox(0,0)[r]{\strut{}\footnotesize $10^{-5}$}}%
      \csname LTb\endcsname%%
      \put(803,2140){\makebox(0,0)[r]{\strut{}\footnotesize $10^{-4}$}}%
      \csname LTb\endcsname%%
      \put(803,2384){\makebox(0,0)[r]{\strut{}\footnotesize $10^{-3}$}}%
      \csname LTb\endcsname%%
      \put(803,2628){\makebox(0,0)[r]{\strut{}\footnotesize $10^{-2}$}}%
      \csname LTb\endcsname%%
      \put(803,2872){\makebox(0,0)[r]{\strut{}\footnotesize $10^{-1}$}}%
      \csname LTb\endcsname%%
      \put(1654,436){\makebox(0,0){\strut{}\footnotesize 400}}%
      \csname LTb\endcsname%%
      \put(2409,436){\makebox(0,0){\strut{}\footnotesize 800}}%
      \csname LTb\endcsname%%
      \put(3163,436){\makebox(0,0){\strut{}\footnotesize 1200}}%
    }%
    \gplgaddtomacro\gplfronttext{%
      \csname LTb\endcsname%%
      \put(2033,76){\makebox(0,0){\strut{}\footnotesize $\omega$}}%
      \csname LTb\endcsname%%
      \put(2033,3232){\makebox(0,0){\strut{}\small Lomb Scargle}}%
    }%
    \gplgaddtomacro\gplbacktext{%
      \csname LTb\endcsname%%
      \put(3289,676){\makebox(0,0)[r]{\strut{}}}%
      \csname LTb\endcsname%%
      \put(3289,920){\makebox(0,0)[r]{\strut{}}}%
      \csname LTb\endcsname%%
      \put(3289,1164){\makebox(0,0)[r]{\strut{}}}%
      \csname LTb\endcsname%%
      \put(3289,1408){\makebox(0,0)[r]{\strut{}}}%
      \csname LTb\endcsname%%
      \put(3289,1652){\makebox(0,0)[r]{\strut{}}}%
      \csname LTb\endcsname%%
      \put(3289,1896){\makebox(0,0)[r]{\strut{}}}%
      \csname LTb\endcsname%%
      \put(3289,2140){\makebox(0,0)[r]{\strut{}}}%
      \csname LTb\endcsname%%
      \put(3289,2384){\makebox(0,0)[r]{\strut{}}}%
      \csname LTb\endcsname%%
      \put(3289,2628){\makebox(0,0)[r]{\strut{}}}%
      \csname LTb\endcsname%%
      \put(3289,2872){\makebox(0,0)[r]{\strut{}}}%
      \csname LTb\endcsname%%
      \put(4140,436){\makebox(0,0){\strut{}\footnotesize 400}}%
      \csname LTb\endcsname%%
      \put(4895,436){\makebox(0,0){\strut{}\footnotesize 800}}%
      \csname LTb\endcsname%%
      \put(5649,436){\makebox(0,0){\strut{}\footnotesize 1200}}%
    }%
    \gplgaddtomacro\gplfronttext{%
      \csname LTb\endcsname%%
      \put(4519,76){\makebox(0,0){\strut{}\footnotesize $\omega$}}%
      \csname LTb\endcsname%%
      \put(4519,3232){\makebox(0,0){\strut{}\small Forward: $\alpha_j = g(x_j)$}}%
    }%
    \gplgaddtomacro\gplbacktext{%
      \csname LTb\endcsname%%
      \put(5775,676){\makebox(0,0)[r]{\strut{}}}%
      \csname LTb\endcsname%%
      \put(5775,920){\makebox(0,0)[r]{\strut{}}}%
      \csname LTb\endcsname%%
      \put(5775,1164){\makebox(0,0)[r]{\strut{}}}%
      \csname LTb\endcsname%%
      \put(5775,1408){\makebox(0,0)[r]{\strut{}}}%
      \csname LTb\endcsname%%
      \put(5775,1652){\makebox(0,0)[r]{\strut{}}}%
      \csname LTb\endcsname%%
      \put(5775,1896){\makebox(0,0)[r]{\strut{}}}%
      \csname LTb\endcsname%%
      \put(5775,2140){\makebox(0,0)[r]{\strut{}}}%
      \csname LTb\endcsname%%
      \put(5775,2384){\makebox(0,0)[r]{\strut{}}}%
      \csname LTb\endcsname%%
      \put(5775,2628){\makebox(0,0)[r]{\strut{}}}%
      \csname LTb\endcsname%%
      \put(5775,2872){\makebox(0,0)[r]{\strut{}}}%
      \csname LTb\endcsname%%
      \put(6626,436){\makebox(0,0){\strut{}\footnotesize 400}}%
      \csname LTb\endcsname%%
      \put(7381,436){\makebox(0,0){\strut{}\footnotesize 800}}%
      \csname LTb\endcsname%%
      \put(8135,436){\makebox(0,0){\strut{}\footnotesize 1200}}%
    }%
    \gplgaddtomacro\gplfronttext{%
      \csname LTb\endcsname%%
      \put(7005,76){\makebox(0,0){\strut{}\footnotesize $\omega$}}%
      \csname LTb\endcsname%%
      \put(7005,3232){\makebox(0,0){\strut{}\small Our estimator}}%
    }%
    \gplbacktext
    \put(0,0){\includegraphics[width={453.00bp},height={170.00bp}]{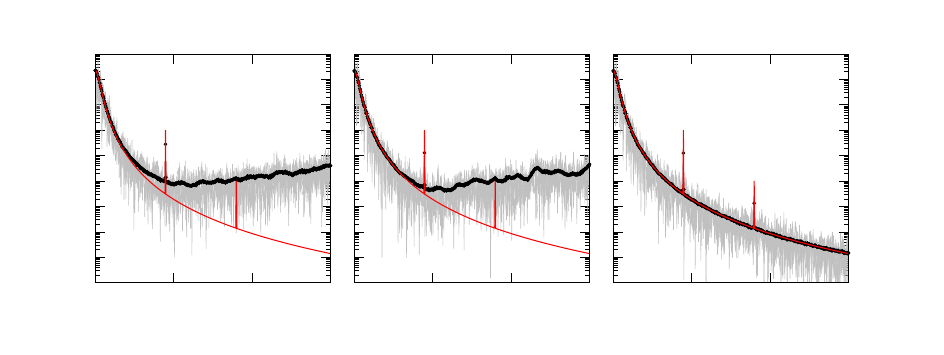}}%
    \gplfronttext
  \end{picture}%
\endgroup
  \caption{A comparison of two existing estimators and the estimator proposed in
  this work, each estimating the spectral density $S$ of a process $Y(x)$ on
  $[0,1]$ with Mat\'ern covariance and two faint spectral lines sampled on a
  jittered grid. Individual estimates for many samples are shown in grey, their
  average is shown in black, and the true SDF $S$ is shown in red. The forward
  estimator is as defined in (\ref{eq:shat}) and $g(x)$ is a standard Kaiser
  window.}
  \label{fig:intro_demo}
\end{figure}

Nonetheless, a large number of approaches have been proposed for estimating
spectral densities in this irregular sampling setting. In the specific case of
identifying discrete periodic components, methods like
ESPRIT~\cite{roy1989esprit} and MUSIC~\cite{schmidt1986multiple} have been
popular and successful (see \cite{babu2010} for a more complete review of
methods in this space). Methods for estimating continuous spectral densities are
similarly numerous, including direct weighted summation~\cite{fuentes2002},
maximum likelihood~\cite{stoica2011maximum}, interpolation-based
methods~\cite{early2020,cui2024fast}, imputation~\cite{guinness2019},
approximate likelihood methods~\cite{matsuda2009,rao2018}, and likelihood-free
methods~\cite{bandyopadhyay2015}. Unfortunately, many of these methods apply
only to data measured on a lattice with gaps. One popular estimator that lifts
this limitation and can be computed for fully irregular data is the
\emph{Lomb-Scargle} periodogram~\cite{lomb1976, scargle1982studies,
vanderplas2018}, which fits a sinusoidal model to the data using least squares
one frequency at a time. Treating regression coefficients at distinct
frequencies as independent estimates of the continuous spectral density in
this way is appealing, as it immediately reduces computational burden and
ostensibly avoids the ill-conditioning inherent in jointly estimation of the
spectral density at multiple frequencies. But as Figure~\ref{fig:intro_demo}
demonstrates, least squares approaches of this variety are far from immune to
bias issues in highly irregular sampling regimes, and aliasing biases can
obscure important spectral features like faint periodic signals.

For gridded data, weighted estimators of the form
\begin{equation} \label{eq:shat}
  \hat{S}(\xi) = \abs{ \sum_{j=1}^n e^{-2 \pi i \xi x_j} \alpha_j \, y_j }^2
\end{equation}
are a foundational tool for spectral analysis.  In this work, we demonstrate
that one can design weights $\set{\alpha_j}_{j=1}^n$ to reduce spectral leakage
even in irregular sampling settings, bringing estimators of the
form~\eqref{eq:shat} for ungridded data closer to parity with their gridded
counterparts. In addition, by leveraging the nonuniform fast Fourier transform
(NUFFT)~\cite{dutt1993,barnett2019}, this estimator can be evaluated at $m$
frequencies $\{\xi_k\}_{k=1}^m$ in $\bO(n + m\log m)$ time, affording some of
the scalability enjoyed by gridded methods which use the FFT. Moreover, we
demonstrate that the NUFFT and related algorithms can be used to accelerate the
computation of weights $\set{\alpha_j}_{j=1}^n$.

Traditionally, weights are designed to be interpretable in the time domain, for
example by selecting $\alpha_j = g(x_j)$ for some \emph{window function} $g(x)$
whose Fourier transform $G(\omg)$ is well-concentrated in a narrow spectral
interval. The key observation we provide in this work is that this design
principle does not carry well to the irregular sampling setting, and instead we
propose to design weights $\set{\alpha_j}_{j=1}^n$ such that
\begin{equation} \label{eq:quad-g}
  H_{\alpha}(\omg) 
  = \sum_{j=1}^n e^{-2 \pi i \omg x_j} \alpha_j
  \approx 
  \int_{a}^b e^{-2 \pi i \omg x} g(x) \dif{x}
  =
  G(\omg).
\end{equation}
For fully gridded data, this design problem is straightforward, and essentially
coincides with the simple $\alpha_j = g(x_j)$ estimator. In the case of
irregular sampling, however, one must carefully specify and numerically solve a
version of~\eqref{eq:quad-g} using tools and perspectives from numerical
quadrature. Figure~\ref{fig:intro_demo} illustrates a significant
improvement in accuracy of our proposed estimator over the na\"ive estimator
with $\alpha_j = g(x_j)$. We now discuss aspects of the design and computation
of the weights $\set{\alpha_j}_{j=1}^n$ that can be done rapidly and to high
accuracy for fully irregular data.

\section{A new spectral density estimator} \label{sec:method}

For expositional clarity, we develop the method here in one dimension. Where
relevant, we provide comments on technical differences in multidimensional
settings, but otherwise the translation of these tools and algorithms to higher
dimensions is direct. A performant software companion to this work which
implements all relevant methods is available at
\texttt{https://github.com/cgeoga/IrregularSpectra.jl}. Proofs for all results
are given in Appendix \ref{app:proofs}.

\subsection{Window functions in standard time series methodology}
\label{sec:windows}

As mentioned above, to reduce finite-sample sources of bias for data being
modeled as truly discrete-time, many nonparametric estimators in this setting
make use of a \emph{window function}, which we will denote $g$, and for the
moment will assume is defined on $[0,1]$ with unit norm $\norm[{L^2([0,1])}]{g}
= 1$ and Fourier transform $G(\omg) = \int_0^1 e^{-2 \pi i \omg x} g(x) \dif x$.
For gridded one-dimensional data sampled on $\set{0, 1, \dots, n-1}$, unlike in
the continuous-time or -space case of~\eqref{eq:sdf}, there is a finite Nyquist
frequency and so the spectral density is modeled on the finite domain as $S :
[-\frac{1}{2},-\frac{1}{2}] \to [0,\infty)$. A popular estimator of this
discrete-time $S$ at frequency $k/n$ for $-n/2 \leq k < n/2$ is
\begin{equation} \label{eq:ts_sdf}
  \hat{S}(k/n) = \abs{
    \frac{1}{\norm[2]{\bm{g}}} \sum_{j=0}^{n-1} e^{-2 \pi i k j/n} g_j y_j
  }^2,
\end{equation}
where $\bm{g} := [g(j/n)]_{j=0}^{n-1}$ is the vector generated by the window
function $g$. A standard computation shows that in this setting we have
\begin{equation} \label{eq:1dsdf_conv}
  \E \hat{S}(k/n) 
  = \int_{-1/2}^{1/2} |G_n(k/n - \omg)|^2 S(\omg) \dif \omg
  = \big(\abs{G_n}^2 * S\big)(k/n)
\end{equation}
where $G_n(\omg) := \frac{1}{\norm[2]{\bm{g}}} \sum_{j=0}^{n-1} e^{-2 \pi i j
\omg} g_j$. Naturally, the ideal window function $g$ would be one that
corresponds to $\abs{G_n(\omg)}^2 \approx \delta(\omega)$, so that
$\hat{S}(k/n)$ would be an unbiased estimator of $S(k/n)$.  For finite $n$, the
Heisenberg uncertainty principle enforces an upper bound on how closely
$\abs{G_n}^2$ can possibly approximate the idealized, perfectly concentrated
delta function~\cite{slepian1983, grochenig2001foundations}. Nevertheless,
designing window functions with favorable concentration properties remains a
driving motivation and important practical concern. The choice $g_j \equiv 1$
implicitly corresponds to selecting the window ${g(x) = \ind{x\in[0,1]}}$, in
which case the estimator~\eqref{eq:ts_sdf} is referred to as the
\emph{periodogram} \citep{brockwell1991}. In this case one obtains
\begin{align} \label{eq:Gn_int}
  n^{-\frac{1}{2}} G_n(k/n) 
  = \frac{1}{n} \sum_{j=0}^{n-1} e^{-2\pi i k j / n}
  &\approx 
  \int_0^1 e^{-2\pi i k x} \dif{x}
  = e^{-ik\pi} \sinc(\pi k)
  = G(k).
\end{align}
The approximate equality above is due to interpreting the summation as an
approximation to the continuous integral using the trapezoidal rule. For
integrands that are smooth and periodic, the trapezoidal rule converges
exponentially~\cite{trefethen2014} and the above approximation is very accurate.
While this is not the typical perspective taken in analyzing window functions,
this quadrature view forms the basis of the new methods proposed here.

\subsection{Extension to irregular sampling} \label{sec:extension}

Let $\by = [Y(x_j)]_{j=1}^n$ denote measurements of the process $Y(x)$ at
potentially non-equidistant locations $\{x_j\}_{j=1}^n$. As discussed above, we
propose to compute estimators of the form~\eqref{eq:shat} where the weights
$\bm{\alpha} \in \R^n$ are designed specifically to satisfy~\eqref{eq:quad-g}
for all $\abs{\omega} \leq \Omega$. This $\Omega$ is the \textit{maximum
controlled frequency}, a user-selected parameter that can be thought of as an
analog to the Nyquist frequency adapted to the context of recovering continuous
window functions from irregular data. Given such weights $\bal$ we may express
the variance of this weighted Fourier sum as a convolution, just as in
\eqref{eq:1dsdf_conv}. Under the assumption of~\eqref{eq:quad-g}, this
convolution may be decomposed as
\begin{align} \label{eq:bias_decomp}
  \E \hat{S}(\xi) 
  &= \int_{\R} |H_{\alpha}(\xi - \omg)|^2 S(\omg) \dif \omg \notag
  \\
  &\approx
  S(\xi) +
  \underbrace{
  \int_{\xi - W}^{\xi + W}
  |G(\xi - \omg)|^2 S(\omg) \dif \omg - S(\xi)
  }_{\text{window-induced bias}} \\ 
  \notag
  &\hspace*{1.35cm} +
  \underbrace{
  \int_{-\infty}^{\xi - \Omega}
  |H_{\alpha}(\xi - \omg)|^2 S(\omg) \dif \omg
  }_{\text{lower aliasing bias}}
  +
  \underbrace{
  \int_{\xi + \Omega}^{\infty}
  |H_{\alpha}(\xi - \omg)|^2 S(\omg) \dif \omg
  }_{\text{upper aliasing bias}}.
\end{align}
The domain in the second term above is reduced because, by assumption on $G$,
the contribution in the intervals $[-\Omega,-W]$ and $[W,\Omega]$ can be made
exceptionally small, for example
\begin{equation*} %\label{eq:}
  \int_{[-\Omega,-W]\cap[W,\Omega]} |G(\xi - \omg)|^2 S(\omg) \dif \omg
  < 10^{-16}
\end{equation*}
for $W \ll \Omega$. If $S$ does not vary wildly on the narrow interval $[\xi-W,
\xi+W]$, then the second term will also be small.  With~\eqref{eq:quad-g}
and~\eqref{eq:bias_decomp} in mind, there are three main design considerations
which must be addressed in order to compute suitable weights $\bm{\alpha}$:
\begin{itemize}
  \item[\textbf{1.}] selecting the maximum controlled frequency $\Omega$ and
  computing the weights $\bal$ so that $H_{\bal}(\omega) \approx G(\omega)$ for
  all $\abs{\omega} \leq \Omega$, 
  \item[\textbf{2.}] choosing the window $g$ so that $G$ is maximally
  concentrated in order to minimize the window-induced bias, and
  \item[\textbf{3.}] balancing the fundamental tradeoff between increasing
  $\Omega$, which permits the estimation of $S$ at higher frequencies, and
  minimizing $\norm[2]{\bal}$, which has a direct relationship with the
  magnitude of the aliasing biases.
\end{itemize}
In the following three sections we treat each of these considerations in turn,
establishing a theoretical basis for the accuracy and limitations of our
estimator.

\subsection{Weight design} \label{sec:weights}

For the moment, let us assume that the window function $g$ is given. We will
return to the problem of selecting $g$ in Section~\ref{sec:window}. The problem
of computing weights $\bm{\alpha}$ which satisfy~\eqref{eq:quad-g} is precisely
the problem of designing a \emph{quadrature rule} with fixed node locations
$\set{x_j}_{j=1}^n$ which accurately integrates the family of functions $\{x
\mapsto e^{-2\pi i\omega x}\}_{\omega\in[-\Omega, \Omega]}$ under the weight
function $g$. If the nodes $\set{x_j}_{j=1}^n$ can be chosen freely, then one
could compute a \textit{Gaussian quadrature rule} which integrates these
functions to high accuracy~\cite{golub1969calculation,yarvin1998generalized}. A
particularly simple choice would be to choose $\set{x_j}_{j=1}^n$ to be the
zeros of the $n$-th order Legendre polynomial, in which case
\emph{Gauss-Legendre} quadrature weights $\set{\gamma_j}_{j=1}^n$ can be
computed and $\bm{\alpha} = \set{\gamma_j g(x_j)}_{j=1}^n$ will
satisfy~\eqref{eq:quad-g}. These nodes and weights can be computed very rapidly
and accurately even for large $n$~\cite{glaser2007fast,hale2013fast}, and so for
practitioners who have the freedom to choose design points, this course of
action provides an immediate and straightforward choice of weights.

\begin{figure}[!t]
  \centering
% GNUPLOT: LaTeX picture with Postscript
\begingroup
  \makeatletter
  \providecommand\color[2][]{%
    \GenericError{(gnuplot) \space\space\space\@spaces}{%
      Package color not loaded in conjunction with
      terminal option `colourtext'%
    }{See the gnuplot documentation for explanation.%
    }{Either use 'blacktext' in gnuplot or load the package
      color.sty in LaTeX.}%
    \renewcommand\color[2][]{}%
  }%
  \providecommand\includegraphics[2][]{%
    \GenericError{(gnuplot) \space\space\space\@spaces}{%
      Package graphicx or graphics not loaded%
    }{See the gnuplot documentation for explanation.%
    }{The gnuplot epslatex terminal needs graphicx.sty or graphics.sty.}%
    \renewcommand\includegraphics[2][]{}%
  }%
  \providecommand\rotatebox[2]{#2}%
  \@ifundefined{ifGPcolor}{%
    \newif\ifGPcolor
    \GPcolortrue
  }{}%
  \@ifundefined{ifGPblacktext}{%
    \newif\ifGPblacktext
    \GPblacktexttrue
  }{}%
  % define a \g@addto@macro without @ in the name:
  \let\gplgaddtomacro\g@addto@macro
  % define empty templates for all commands taking text:
  \gdef\gplbacktext{}%
  \gdef\gplfronttext{}%
  \makeatother
  \ifGPblacktext
    % no textcolor at all
    \def\colorrgb#1{}%
    \def\colorgray#1{}%
  \else
    % gray or color?
    \ifGPcolor
      \def\colorrgb#1{\color[rgb]{#1}}%
      \def\colorgray#1{\color[gray]{#1}}%
      \expandafter\def\csname LTw\endcsname{\color{white}}%
      \expandafter\def\csname LTb\endcsname{\color{black}}%
      \expandafter\def\csname LTa\endcsname{\color{black}}%
      \expandafter\def\csname LT0\endcsname{\color[rgb]{1,0,0}}%
      \expandafter\def\csname LT1\endcsname{\color[rgb]{0,1,0}}%
      \expandafter\def\csname LT2\endcsname{\color[rgb]{0,0,1}}%
      \expandafter\def\csname LT3\endcsname{\color[rgb]{1,0,1}}%
      \expandafter\def\csname LT4\endcsname{\color[rgb]{0,1,1}}%
      \expandafter\def\csname LT5\endcsname{\color[rgb]{1,1,0}}%
      \expandafter\def\csname LT6\endcsname{\color[rgb]{0,0,0}}%
      \expandafter\def\csname LT7\endcsname{\color[rgb]{1,0.3,0}}%
      \expandafter\def\csname LT8\endcsname{\color[rgb]{0.5,0.5,0.5}}%
    \else
      % gray
      \def\colorrgb#1{\color{black}}%
      \def\colorgray#1{\color[gray]{#1}}%
      \expandafter\def\csname LTw\endcsname{\color{white}}%
      \expandafter\def\csname LTb\endcsname{\color{black}}%
      \expandafter\def\csname LTa\endcsname{\color{black}}%
      \expandafter\def\csname LT0\endcsname{\color{black}}%
      \expandafter\def\csname LT1\endcsname{\color{black}}%
      \expandafter\def\csname LT2\endcsname{\color{black}}%
      \expandafter\def\csname LT3\endcsname{\color{black}}%
      \expandafter\def\csname LT4\endcsname{\color{black}}%
      \expandafter\def\csname LT5\endcsname{\color{black}}%
      \expandafter\def\csname LT6\endcsname{\color{black}}%
      \expandafter\def\csname LT7\endcsname{\color{black}}%
      \expandafter\def\csname LT8\endcsname{\color{black}}%
    \fi
  \fi
    \setlength{\unitlength}{0.0500bp}%
    \ifx\gptboxheight\undefined%
      \newlength{\gptboxheight}%
      \newlength{\gptboxwidth}%
      \newsavebox{\gptboxtext}%
    \fi%
    \setlength{\fboxrule}{0.5pt}%
    \setlength{\fboxsep}{1pt}%
    \definecolor{tbcol}{rgb}{1,1,1}%
\begin{picture}(8500.00,3400.00)%
    \gplgaddtomacro\gplbacktext{%
      \csname LTb\endcsname%%
      \put(747,813){\makebox(0,0)[r]{\strut{}\footnotesize 0}}%
      \csname LTb\endcsname%%
      \put(747,1476){\makebox(0,0)[r]{\strut{}\footnotesize 0.001}}%
      \csname LTb\endcsname%%
      \put(747,2138){\makebox(0,0)[r]{\strut{}\footnotesize 0.002}}%
      \csname LTb\endcsname%%
      \put(1112,436){\makebox(0,0){\strut{}\footnotesize -0.8}}%
      \csname LTb\endcsname%%
      \put(1642,436){\makebox(0,0){\strut{}\footnotesize -0.4}}%
      \csname LTb\endcsname%%
      \put(2173,436){\makebox(0,0){\strut{}\footnotesize 0}}%
      \csname LTb\endcsname%%
      \put(2703,436){\makebox(0,0){\strut{}\footnotesize 0.4}}%
      \csname LTb\endcsname%%
      \put(3233,436){\makebox(0,0){\strut{}\footnotesize 0.8}}%
    }%
    \gplgaddtomacro\gplfronttext{%
      \csname LTb\endcsname%%
      \put(2172,76){\makebox(0,0){\strut{}\small $x$}}%
      \csname LTb\endcsname%%
      \put(2172,3063){\makebox(0,0){\strut{}\small weight values $\alpha_j$}}%
    }%
    \gplgaddtomacro\gplbacktext{%
      \csname LTb\endcsname%%
      \put(4245,1014){\makebox(0,0)[r]{\strut{}\footnotesize $10^{-20}$}}%
      \csname LTb\endcsname%%
      \put(4245,1436){\makebox(0,0)[r]{\strut{}\footnotesize $10^{-15}$}}%
      \csname LTb\endcsname%%
      \put(4245,1858){\makebox(0,0)[r]{\strut{}\footnotesize $10^{-10}$}}%
      \csname LTb\endcsname%%
      \put(4245,2281){\makebox(0,0)[r]{\strut{}\footnotesize $10^{-5}$}}%
      \csname LTb\endcsname%%
      \put(4245,2703){\makebox(0,0)[r]{\strut{}\footnotesize $10^{0}$}}%
      \csname LTb\endcsname%%
      \put(4345,436){\makebox(0,0){\strut{}\footnotesize 0}}%
      \csname LTb\endcsname%%
      \put(4769,436){\makebox(0,0){\strut{}\footnotesize 20}}%
      \csname LTb\endcsname%%
      \put(5193,436){\makebox(0,0){\strut{}\footnotesize 40}}%
      \csname LTb\endcsname%%
      \put(5617,436){\makebox(0,0){\strut{}\footnotesize 60}}%
      \csname LTb\endcsname%%
      \put(6041,436){\makebox(0,0){\strut{}\footnotesize 80}}%
      \csname LTb\endcsname%%
      \put(6465,436){\makebox(0,0){\strut{}\footnotesize 100}}%
      \csname LTb\endcsname%%
      \put(6889,436){\makebox(0,0){\strut{}\footnotesize 120}}%
    }%
    \gplgaddtomacro\gplfronttext{%
      \csname LTb\endcsname%%
      \put(8188,2584){\makebox(0,0)[r]{\strut{}\footnotesize $|G(\omg)|^2$}}%
      \csname LTb\endcsname%%
      \put(8188,2344){\makebox(0,0)[r]{\strut{}\footnotesize no correction}}%
      \csname LTb\endcsname%%
      \put(8188,2104){\makebox(0,0)[r]{\strut{}\footnotesize irreg. trap.}}%
      \csname LTb\endcsname%%
      \put(8188,1864){\makebox(0,0)[r]{\strut{}\footnotesize our method}}%
      \csname LTb\endcsname%%
      \put(5670,76){\makebox(0,0){\strut{}\small $\omg$}}%
      \csname LTb\endcsname%%
      \put(5670,3063){\makebox(0,0){\strut{}\small $|\sum_{j=1}^n e^{-2 \pi i \omg x_j} \alpha_j|^2$}}%
    }%
    \gplbacktext
    \put(0,0){\includegraphics[width={425.00bp},height={170.00bp}]{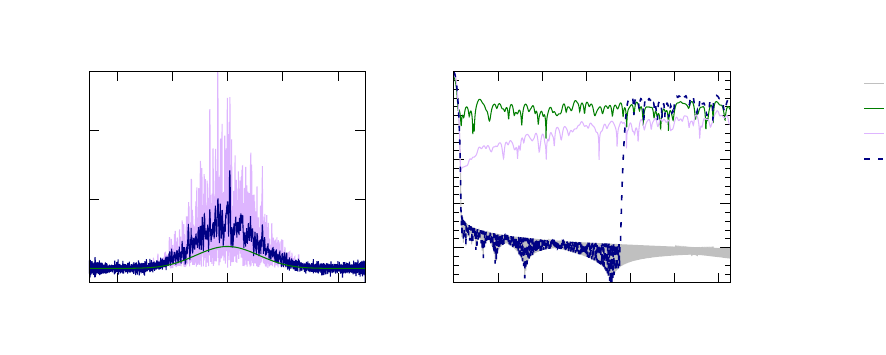}}%
    \gplfronttext
  \end{picture}%
\endgroup
  \caption{Accuracy of various weighting schemes in recovering $G$ using a
  weighted Fourier sum with $\set{x_j}_{j=1}^n \iid \text{Unif}([-1, 1])$.
  Weights for our method are computed with $\Omega = 75$. The window $g$ is
  chosen to be a Kaiser function, discussed in the next subsection, and the
  weights are $\alpha_j = g(x_j)/n$ for no correction and $\alpha_j = g(x_j)
  \gamma_j$ as in~\eqref{eq:irtrap} for the irregular trapezoidal scheme.}
  \label{fig:flat_trap_ours}
\end{figure}

In the much more common case in which data are provided at arbitrary irregular
locations $\set{x_j}_{j=1}^n$ without input from the analyst, no such high-order
quadrature exists a priori. One particularly simple choice of low-order
quadrature rule in this setting is an irregular trapezoidal rule with weights 
\begin{equation} \label{eq:irtrap} 
  \gamma_j = \begin{cases}
    (x_2 - x_1) / 2 & \text{for} \ j=1 \\
    (x_{j+1} - x_{j-1}) / 2 & \text{for} \ 1 < j < n \\
    (x_n - x_{n-1}) / 2 & \text{for} \ j=n.
  \end{cases}
\end{equation}
These weights are sometimes referred to as \textit{density compensation factors}
in imaging inverse
problems~\cite{grochenig1993discrete,viswanathan2010reconstruction,adcock2014stable},
and may be satisfactory for estimating spectral densities at low to moderate
frequencies from densely-sampled or quasi-equispaced data. However, as is
demonstrated in Figure \ref{fig:flat_trap_ours}, even mild gaps in the data lead
to large errors when integrating sufficiently oscillatory functions.

In order to develop an alternative quadrature rule on irregular nodes
$\{x_j\}_{j=1}^n$ which gives higher accuracy, we turn \eqref{eq:quad-g} into a
linear system which can be solved for the weights $\bal$. For this purpose, we
take $\omega_k := \Omega\cos(\frac{2k - 1}{2n} \pi)$ for $k=1,\dots,n$ to be
Chebyshev nodes on $[-\Omega,\Omega]$. Let $\Frt \in \C^{n \times n}$ denote the
nonuniform discrete Fourier matrix given by $\Frt_{jk} := e^{-2 \pi i \omg_j
x_k}$ and define $\bm{b} := [G(\omg_j)]_{j=1}^n$. Then we seek solutions
$\bm{\alpha}$ to the linear system 
\begin{equation} \label{eq:weights_linsys} 
  \Frt \bm{\alpha}
= \bm{b}. 
\end{equation} 
The matrix $\bm{\Frt}$ is exceptionally ill-conditioned in most
cases~\cite{pan2016bad}. However, by using the fact that $G$ is entire, the
following theorem demonstrates that with mild oversampling, a solution
$\bm{\alpha}$ to the linear system~\eqref{eq:weights_linsys}---and thus a
suitable choice of weights satisfying~\eqref{eq:quad-g}---is guaranteed to
exist. 
\begin{theorem} \label{thm:recovery} Take $\Omega > 0$ and $a \leq x_1 < \dots <
  x_n \leq b$. Let $G$ be bandlimited with $\supp(g) \subseteq [a,b]$. Then for
  any $\rho > 1$ there exist weights $\bm{\alpha} \in \R^n$ such that
  \begin{equation*}
    \norm[{L^\infty([-\Omega,\Omega])}]{H_{\bm{\alpha}} - G}
    \leq \frac{4}{\rho - 1} \exp\left\{\frac{\pi}{2} \Omega (b-a) (\rho + \rho^{-1}) - n \log\rho \right\} \left( \norm[{L^1([a,b])}]{g} + \norm[1]{\bm{\alpha}} \right).
  \end{equation*}
\end{theorem}

Note that we can choose the rate of decay $\rho$ in Theorem~\ref{thm:recovery}
to be arbitrarily large, at the cost of an increased multiplicative constant.
Therefore, this theorem demonstrates that the approximation error between
$H_{\bal}(\omega)$ and $G(\omega)$ on the interval $\omega \in [-\Omega,\Omega]$
decays \textit{superexponentially} in $n$ (i.e. faster than any exponential),
with a rate and prefactor which depend on the space-frequency product
$2\Omega(b-a)$ and the choice of $\rho$.

To make this result more concrete, we provide the following corollary, which
fixes $\rho$, demonstrating the oversampling necessary to accurately recover the
window $G$ on an interval $[-\Omega,\Omega]$.

\begin{corollary} \label{cor:recovery} Take $\epsilon > 10^{-16}$ and $n > 0$.
Define the pseudo-Nyquist frequency $\Omega_{\text{\normalfont nyq}} :=
\frac{n}{2(b-a)}$. Then under the assumptions of Theorem \ref{thm:recovery}, for
any $\Omega$ such that
  \begin{equation*}
    \Omega 
    \leq \frac{n - 35}{5(b-a)}
    = \frac{2}{5} \, \Omega_{\text{\normalfont nyq}} - \frac{7}{b-a},
  \end{equation*}
  there exist weights $\bm{\alpha} \in \R^n$ such that
  \begin{equation*}
    \norm[{L^\infty([-\Omega,\Omega])}]{H_{\bm{\alpha}} - G}
    \leq \epsilon \left( \norm[{L^1([a,b])}]{g} + \norm[1]{\bm{\alpha}} \right)
  \end{equation*}
\end{corollary}
\begin{proof}
  Take $\rho = 3$ in Theorem \ref{thm:recovery}, set the right hand side equal
  to $\epsilon( \norm[{L^1([a,b])}]{g} + \norm[1]{\bm{\alpha}} )$, and use the
  facts $\frac{5\pi}{3\log 3} < 5$ and $-\log_3\big(\frac{10^{-16}}{2}\big) <
  35$.
\end{proof}

Corollary~\ref{cor:recovery} roughly states that the window recovery problem can
be solved to 16 digit accuracy relative to $\norm[1]{\bal}$ for any $\Omega \leq
\frac{2}{5} \, \Omega_{\text{nyq}}$, where the \textit{pseudo-Nyquist} frequency
$\Omega_{\text{nyq}}$ is simply the Nyquist frequency if $\set{x_j}_{j=1}^n$
were equispaced on $[a,b]$. We will see in Section~\ref{sec:window}, however,
that this may result in prohibitively large oscillatory weights depending on the
choice of $g$.  By fixing the interval $[a,b]$ on which observations are
collected and letting $n \to \infty$, applying Corollary~\ref{cor:recovery} at
each $n$ also establishes the existence of a sequence $\Omega_n \to \infty$ for
which the recovery problem can we solved to high accuracy. This observation
implies a notion of asymptotic convergence, which we treat in more detail in
Corollary~\ref{cor:rate}.

It is worth emphasizing that the window reconstruction problem differs in
several important ways from the well-studied \textit{irregular sampling} problem
of computing Fourier coefficients from data by solving the inverse problem
$\bm{\Frt}^* \bm{c} = \bm{y}$~\cite{benedetto1992irregular,
aldroubi2001nonuniform,adcock2014stable,grochenig1993discrete}. Computing
weights $\bm{\alpha}$ which satisfy $\bm{\Frt} \bal = \bm{b}$ when $G$ is a
smooth, bandlimited, concentrated function is possible even when the general
inverse problem is hopelessly ill-conditioned, and no such weights could be
stably computed for non-smooth, noisy data. In the irregular sampling problem
$\bm{\Frt}^* \bm{c} = \bm{y}$, clustered points $x_j$ induce ill-conditioning as
small perturbations to the samples $y_j$ can lead to wild oscillations in the
recovered coefficients $\bm{c}$. However, in our window reconstruction setting,
the existence of $\bm{\alpha}$ requires only that the window function
$G(\omega)$ lies approximately in $\spn(\{e^{-2\pi i\omega x_j}\}_{j=1}^n)$ for
$\omega \in [-\Omega, \Omega]$, from which one can compute accurate weights
$\bm{\alpha}$ using any backwards stable method of solving the least squares
problem $\bm{\Fr} \bm{\alpha} = \bm{b}$, for example a Householder QR
factorization~\cite{trefethen2022numerical}. Therefore adding sampling points
$x_j$ can only improve the accuracy of the reconstructed window
$H_{\bal}(\omega)$.

\subsection{Choosing a window} \label{sec:window}

We now turn to the problem of designing a suitable window function $g$. As has
been discussed already, the primary quality of a window function is that its
Fourier transform $G(\omg) = \int_a^b g(x) e^{-2 \pi i \omg x} \dif x$ is a
close approximation to a delta mass. This intuitive statement can be interpreted
in multiple ways, and thus several tradeoffs have been explored at length in the
time series and signal processing literatures between making $G$ more narrowly
peaked at the origin (reducing ``local" bias) and reducing its mass off of a
main lobe (reducing ``broadband" bias). In certain applications, such as picking
out sharp nearby peaks in $S(\omg)$, choices for $g$ like the \emph{sine tapers}
that minimize local bias~\cite{riedel1995} can be a good choice and are
compatible with our method. But in most cases, a better default is to elect to
minimize broadband bias, and so we now provide a more detailed description of
window design for optimal concentration around the origin using the
\emph{Slepian} or \emph{prolate} functions, first described in
\cite{slepian1961} and which serve as the backbone of the celebrated multitaper
estimator in time series analysis \cite{thomson1982}.

Assuming momentarily that $g$ is supported on $[-\frac{1}{2}, \frac{1}{2}]$ for
notational convenience and normalized so that $\norm[{L^2([-\frac{1}{2},
\frac{1}{2}])}]{g} = 1$, the \emph{spectral concentration} of $g$ on $[-W, W]$
is given by
\begin{equation} \label{eq:conc}
  \lam(g) := \int_{-W}^W |G(\omg)|^2 \dif \omg.
\end{equation}
We note that $\lam(g) < 1$, and the closer this concentration is to one, the
closer $g$ is to being a truly bandlimited function (which is not possible for
any function with finite support). In~\cite{slepian1961}, the authors prove that
the \emph{optimally} concentrated function in the sense of maximizing $\lambda$
is the dominant eigenfunction of the linear integral operator
\begin{equation*} % \label{eq:fredslep}
  T: f(x) \mapsto \int_{-1/2}^{1/2} \text{sinc}(W(t - x)) f(t) \dif t,
\end{equation*}
and \cite{simons2011} gives a more general convenient form of this operator in
multiple dimensions for functions $f(\bx)$ defined on domains $\mathcal{D}$ and
with arbitrary regions of concentration $\mathcal{R}$ given by
\begin{equation} \label{eq:fredslep_highd}
  T: f(\bx) \mapsto \int_{\mathcal{D}} \set{ \int_{\mathcal{R}} e^{2 \pi i
  \bm{s}^T (\bm{t} - \bx)} \dif \bm{s}} f(\bm{t}) \dif \bm{t}.
\end{equation}
Such optimally concentrated functions are called \emph{prolate spheroidal
wavefunctions} (PSWFs), or simply as prolates. Many fast and accurate numerical
methods and algorithms for evaluating prolates
exist~\cite{xiao2001,gruenbacher2002,greengard2024}, and the implementation used
here computes them in $\bO(n \log n)$ time by discretizing $T$ using
Gauss-Legendre quadrature and accelerating the action of the discretized
integral operator using the \emph{fast sinc transform} \cite{greengard2007}.
With this acceleration, obtaining a small number of dominant eigenvectors using
implicit Krylov methods can be done in quasilinear runtime cost even without a
preconditioner, as the operator $T$ has a small number of clustered
non-degenerate eigenvalues but then exhibits rapid spectral decay.

As a simple option for data observed on an interval in one dimension, an
approximation to the prolate is given in closed form by the \emph{Kaiser} window
\citep{kaiser1980}. For $[a,b] = [-\frac{1}{2},\frac{1}{2}]$, this window is
given in the time and Fourier domain by
\begin{equation} \label{eq:kaiser}
  g(x) = \ind{x \in [-1/2, 1/2]} c_0 I_0(\beta \sqrt{1 - (2x)^2})
  \quad \quad
  G(\omg) = \begin{cases}
    \frac{c_0 \sinh(\sqrt{\beta^2 - (2 \pi \omg)^2})}{\sqrt{\beta^2 - (2 \pi \omg)^2}} & |2 \pi \omg| < \beta
    \\
    c_0 \text{sinc}(
      \sqrt{(2 \pi \omg)^2 - \beta^2}) & |2 \pi \omg| \geq \beta
  \end{cases},
\end{equation}
where $I_0$ is the first-kind modified Bessel function, $\beta$ is a shape
parameter, and $c_0$ is selected so that
$\norm[{L^2([-\frac{1}{2},\frac{1}{2}])}]{g} = 1$~\cite{nist}. It is
straightforward to shift and rescale this Fourier transform pair so that $g$ is
supported on an arbitrary interval $[a, b]$, and we will not comment in detail
on such transformations when they are done. While the prolate function is the
theoretically superior choice, the Kaiser window is fast and convenient to
evaluate and gives comparable performance in practice, and thus we utilize it in
a number of numerical demonstrations here.

For the purposes of our estimator, it is also important that the window function
$g$ is smooth and only supported on domains containing the sampling locations.
The motivation for these soft requirements is more subtle than the first and
pertains to controlling the norm of $\bal$. As Theorem~\ref{thm:recovery}
indicates, if $\norm[1]{\bal}$ is large, then the absolute error
$\norm[{L^\infty([a,b])}]{H_{\bal} - G}$ in the window reconstruction will in
general also be large. In addition, we will see in Section~\ref{sec:analysis}
that $\norm[2]{\bal}$ controls the size of the aliasing biases. To provide some
intuition for the importance of smoothness and matching the support of the data,
we return to the linear system $\bm{\Frt} \bal = \bm{b}$. Consider taking an SVD
to obtain $\bm{U} \bm{D} \bm{V}^* = \bm{\Frt}$, where $\bm{U}, \bm{V} \in \C^{n
\times n}$ are unitary matrices, and $\bm{D} \in \R^{n \times n}$ is a diagonal
matrix whose entries are the singular values $\{\sigma_j\}_{j=1}^n$ of
$\bm{\Frt}$. Since the $\ell_2$ norm is invariant to unitary transformations, we
have $\norm[2]{\bal} = \norm[2]{\bm{D}^{-1} \bm{U}^* \bm{b}}$. As previously
noted, the nonuniform Fourier matrix $\bm{\Frt}$ is typically highly
ill-conditioned, and thus its singular values decay rapidly. If $\bm{b}$ lies
approximately in the span of the dominant $r$ eigenvectors of $\bm{\Frt}$, then
$(\bm{U}^* \bm{b})_j \approx 0$ for $j > r$, and thus $(\bm{U}^* \bm{b})_j /
\sigma_j$ remains small in magnitude even for very small singular values
$\sigma_j$. In contrast, if $\bm{b}$ is not orthogonal to the singular vectors
corresponding to small singular values $\sigma_j$, then $(\bm{U}^* \bm{b})_j /
\sigma_j \gg 1$ for some $j$, and $\norm[2]{\bal}$ will also be large. If $g$ is
large in some interval where no sampling locations lie, then we are attempting
to approximate $G$ by sinusoids whose dominant columnspace does not include the
relevant frequencies. Similarly, if $g$ contains discontinuities in its
derivatives, then $G$ will decay slowly, and more sinusoids are necessary to
capture its behavior on its larger domain of numerical support. For a smooth $g$
which takes large values only in regions where sampling locations are present,
we avoid these numerical issues, so that $G$ and the corresponding $\bm{b}$ lie
handily in the span of the dominant eigenvectors and $\norm[2]{\bal}$ is small. 

\begin{figure}[!ht]
  \centering
% GNUPLOT: LaTeX picture with Postscript
\begingroup
  \makeatletter
  \providecommand\color[2][]{%
    \GenericError{(gnuplot) \space\space\space\@spaces}{%
      Package color not loaded in conjunction with
      terminal option `colourtext'%
    }{See the gnuplot documentation for explanation.%
    }{Either use 'blacktext' in gnuplot or load the package
      color.sty in LaTeX.}%
    \renewcommand\color[2][]{}%
  }%
  \providecommand\includegraphics[2][]{%
    \GenericError{(gnuplot) \space\space\space\@spaces}{%
      Package graphicx or graphics not loaded%
    }{See the gnuplot documentation for explanation.%
    }{The gnuplot epslatex terminal needs graphicx.sty or graphics.sty.}%
    \renewcommand\includegraphics[2][]{}%
  }%
  \providecommand\rotatebox[2]{#2}%
  \@ifundefined{ifGPcolor}{%
    \newif\ifGPcolor
    \GPcolortrue
  }{}%
  \@ifundefined{ifGPblacktext}{%
    \newif\ifGPblacktext
    \GPblacktexttrue
  }{}%
  % define a \g@addto@macro without @ in the name:
  \let\gplgaddtomacro\g@addto@macro
  % define empty templates for all commands taking text:
  \gdef\gplbacktext{}%
  \gdef\gplfronttext{}%
  \makeatother
  \ifGPblacktext
    % no textcolor at all
    \def\colorrgb#1{}%
    \def\colorgray#1{}%
  \else
    % gray or color?
    \ifGPcolor
      \def\colorrgb#1{\color[rgb]{#1}}%
      \def\colorgray#1{\color[gray]{#1}}%
      \expandafter\def\csname LTw\endcsname{\color{white}}%
      \expandafter\def\csname LTb\endcsname{\color{black}}%
      \expandafter\def\csname LTa\endcsname{\color{black}}%
      \expandafter\def\csname LT0\endcsname{\color[rgb]{1,0,0}}%
      \expandafter\def\csname LT1\endcsname{\color[rgb]{0,1,0}}%
      \expandafter\def\csname LT2\endcsname{\color[rgb]{0,0,1}}%
      \expandafter\def\csname LT3\endcsname{\color[rgb]{1,0,1}}%
      \expandafter\def\csname LT4\endcsname{\color[rgb]{0,1,1}}%
      \expandafter\def\csname LT5\endcsname{\color[rgb]{1,1,0}}%
      \expandafter\def\csname LT6\endcsname{\color[rgb]{0,0,0}}%
      \expandafter\def\csname LT7\endcsname{\color[rgb]{1,0.3,0}}%
      \expandafter\def\csname LT8\endcsname{\color[rgb]{0.5,0.5,0.5}}%
    \else
      % gray
      \def\colorrgb#1{\color{black}}%
      \def\colorgray#1{\color[gray]{#1}}%
      \expandafter\def\csname LTw\endcsname{\color{white}}%
      \expandafter\def\csname LTb\endcsname{\color{black}}%
      \expandafter\def\csname LTa\endcsname{\color{black}}%
      \expandafter\def\csname LT0\endcsname{\color{black}}%
      \expandafter\def\csname LT1\endcsname{\color{black}}%
      \expandafter\def\csname LT2\endcsname{\color{black}}%
      \expandafter\def\csname LT3\endcsname{\color{black}}%
      \expandafter\def\csname LT4\endcsname{\color{black}}%
      \expandafter\def\csname LT5\endcsname{\color{black}}%
      \expandafter\def\csname LT6\endcsname{\color{black}}%
      \expandafter\def\csname LT7\endcsname{\color{black}}%
      \expandafter\def\csname LT8\endcsname{\color{black}}%
    \fi
  \fi
    \setlength{\unitlength}{0.0500bp}%
    \ifx\gptboxheight\undefined%
      \newlength{\gptboxheight}%
      \newlength{\gptboxwidth}%
      \newsavebox{\gptboxtext}%
    \fi%
    \setlength{\fboxrule}{0.5pt}%
    \setlength{\fboxsep}{1pt}%
    \definecolor{tbcol}{rgb}{1,1,1}%
\begin{picture}(9060.00,3400.00)%
    \gplgaddtomacro\gplbacktext{%
      \csname LTb\endcsname%%
      \put(351,817){\makebox(0,0)[r]{\strut{}\footnotesize 0}}%
      \csname LTb\endcsname%%
      \put(351,1289){\makebox(0,0)[r]{\strut{}\footnotesize 0.5}}%
      \csname LTb\endcsname%%
      \put(351,1760){\makebox(0,0)[r]{\strut{}\footnotesize 1}}%
      \csname LTb\endcsname%%
      \put(351,2232){\makebox(0,0)[r]{\strut{}\footnotesize 1.5}}%
      \csname LTb\endcsname%%
      \put(351,2703){\makebox(0,0)[r]{\strut{}\footnotesize 2}}%
      \csname LTb\endcsname%%
      \put(452,436){\makebox(0,0){\strut{}\footnotesize -1}}%
      \csname LTb\endcsname%%
      \put(1299,436){\makebox(0,0){\strut{}\footnotesize -0.5}}%
      \csname LTb\endcsname%%
      \put(2146,436){\makebox(0,0){\strut{}\footnotesize 0}}%
      \csname LTb\endcsname%%
      \put(2994,436){\makebox(0,0){\strut{}\footnotesize 0.5}}%
      \csname LTb\endcsname%%
      \put(3841,436){\makebox(0,0){\strut{}\footnotesize 1}}%
    }%
    \gplgaddtomacro\gplfronttext{%
      \csname LTb\endcsname%%
      \put(2146,76){\makebox(0,0){\strut{}\small $x$}}%
      \csname LTb\endcsname%%
      \put(2146,3063){\makebox(0,0){\strut{}\small Time-domain windows}}%
    }%
    \gplgaddtomacro\gplbacktext{%
      \csname LTb\endcsname%%
      \put(4645,676){\makebox(0,0)[r]{\strut{}\footnotesize $10^{-25}$}}%
      \csname LTb\endcsname%%
      \put(4645,1066){\makebox(0,0)[r]{\strut{}\footnotesize $10^{-20}$}}%
      \csname LTb\endcsname%%
      \put(4645,1456){\makebox(0,0)[r]{\strut{}\footnotesize $10^{-15}$}}%
      \csname LTb\endcsname%%
      \put(4645,1845){\makebox(0,0)[r]{\strut{}\footnotesize $10^{-10}$}}%
      \csname LTb\endcsname%%
      \put(4645,2235){\makebox(0,0)[r]{\strut{}\footnotesize $10^{-5}$}}%
      \csname LTb\endcsname%%
      \put(4645,2625){\makebox(0,0)[r]{\strut{}\footnotesize $10^{0}$}}%
      \csname LTb\endcsname%%
      \put(4746,436){\makebox(0,0){\strut{}\footnotesize 0}}%
      \csname LTb\endcsname%%
      \put(5424,436){\makebox(0,0){\strut{}\footnotesize 10}}%
      \csname LTb\endcsname%%
      \put(6102,436){\makebox(0,0){\strut{}\footnotesize 20}}%
      \csname LTb\endcsname%%
      \put(6779,436){\makebox(0,0){\strut{}\footnotesize 30}}%
      \csname LTb\endcsname%%
      \put(7457,436){\makebox(0,0){\strut{}\footnotesize 40}}%
      \csname LTb\endcsname%%
      \put(8135,436){\makebox(0,0){\strut{}\footnotesize 50}}%
    }%
    \gplgaddtomacro\gplfronttext{%
      \csname LTb\endcsname%%
      \put(7637,2488){\makebox(0,0)[r]{\strut{}\footnotesize prolate ($\norm[2]{\bal} \approx 0.01$)}}%
      \csname LTb\endcsname%%
      \put(7637,2248){\makebox(0,0)[r]{\strut{}\footnotesize Kaiser ($\norm[2]{\bal} \approx 10^{11}$)}}%
      \csname LTb\endcsname%%
      \put(6440,76){\makebox(0,0){\strut{}\small $\omg$}}%
      \csname LTb\endcsname%%
      \put(6440,3063){\makebox(0,0){\strut{}\small $|H_{\alpha}(\omg)|^2$}}%
    }%
    \gplbacktext
    \put(0,0){\includegraphics[width={453.00bp},height={170.00bp}]{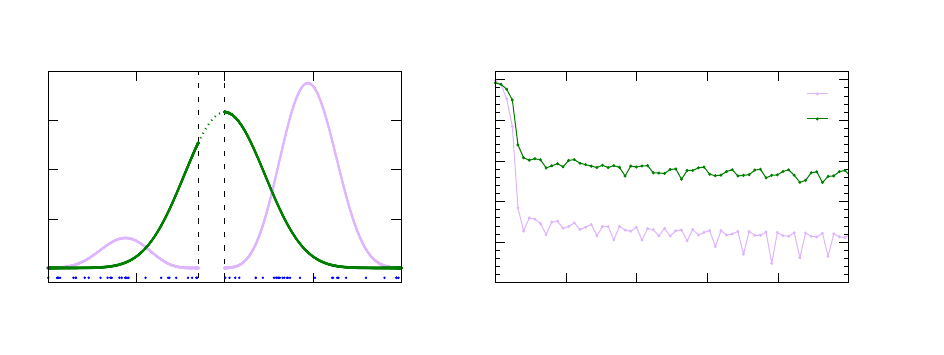}}%
    \gplfronttext
  \end{picture}%
\endgroup
  \caption{A comparison of the recovered window $\abs{H_{\alpha}}^2$ in a
  sampling setting with $\set{x_j}_{j=1}^n \iid \text{Unif}([-1,
  1]\setminus[-0.15,0])$ and $n = 3000$ using a generic Kaiser window supported
  on $[-1, 1]$ and a prolate window supported on the union of disjoint
  intervals. Blue dots in the left subplot show a representative subset of
  sampling locations.}
  \label{fig:gap_prolate}
\end{figure}

Figure \ref{fig:gap_prolate} provides a visual demonstration of this phenomenon.
Data locations are chosen uniformly at random on $[-1, 1] \setminus [-0.15, 0]$,
meaning that there will necessarily be a gap in the sorted points of size at
least $0.15$. While the recovered weights $\bal$ for window $g$ selected to be a
standard Kaiser window on $[-1, 1]$ do still provide a reasonably
well-concentrated approximation $H_{\alpha}(\omg)$ (accounting for the loss of a
few digits due to the size of $\norm[2]{\bm{\alpha}}$), the weights have a norm
on the order of $10^{11}$. Using a prolate that has been adapted specifically to
the data domain $\mathcal{D} = [-1, 1] \setminus [-0.15, 0]$, on the other hand,
provides weights with a norm $13$ orders of magnitude smaller and concentration
at the level of machine precision. 

Finally, we note that the problem of choosing $H_{\bm{\beta}}(\omega) =
\sum_{j=1}^n \beta_j e^{-2 \pi i \omg x_j}$ to be maximally concentrated in a
window $[-W, W]$ can be approached from an optimization rather than a quadrature
perspective. In particular, selecting an optimally-concentrated
$H_{\bm{\beta}}(\omega)$ can be formulated as
\begin{align} \label{eq:GPSS-optimization}
  \max_{\bbeta \in \R^n} \norm[{L^2([-W, W])}]{\sum_{j=1}^n \beta_j e^{-2 \pi i \omg x_j}}
  \quad
  \text{s.t.}
  \quad
  \quad
  \norm[{L^2([-\Omega, \Omega])}]{\sum_{j=1}^n \beta_j e^{-2 \pi i \omg x_j}}=1.
\end{align}
We can express these $L^2$ norms exactly as vector norms by defining the
matrices
\begin{equation} \label{eq:GEP-matrices}
  \bm{A}_{jk} = \frac{\sin(2 \pi W(x_j - x_k))}{\pi (x_j - x_k)}
  \quad\quad
  \text{and}
  \quad\quad
  \bm{B}_{jk} = \frac{\sin(2 \pi \Omega(x_j - x_k))}{\pi (x_j - x_k)}
\end{equation}
and noting that $\shortnorm[{L^2([-W, W])}]{\sum_{j=1}^n \beta_j e^{-2 \pi i
\omg x_j}} = \bbeta^* \bm{A} \bbeta$ and $\shortnorm[{L^2([-\Omega,
\Omega])}]{\sum_{j=1}^n \beta_j e^{-2 \pi i \omg x_j}} = \bbeta^* \bm{B}
\bbeta$. Therefore, taking $\bbeta$ to solve the optimization
problem~\eqref{eq:GPSS-optimization} is equivalent to choosing $\bbeta$ to be
the dominant eigenvector of the generalized eigenvalue problem (GEP) $\bm{A}
\bbeta = \lambda \bm{B} \bbeta$, normalized so that $\bbeta^* \bm{B} \bbeta =
1$. This is exactly the form of the weights $\bbeta$ found in the Bronez
estimator~\cite{bronez1988}, in which $\bbeta$ are referred to as
\emph{generalized} prolate sequences. Bronez proposes an algorithm of complexity
$\bO(n^4)$ which solves a GEP for every desired frequency $\xi$ to obtain
estimates $\hat{S}(\xi)$ in the case of bandlimited processes with
$\text{supp}(S) \subseteq [-\Omega, \Omega]$.

Several comments are worth making on this comparison. First, the GPSS weights
described above are specified by the observation locations alone and do not
require a geometric description of the sampling domain. This makes the GPSS
problem easier to specify than our approach, which may yield less biased
spectral estimates in certain sampling regimes compared to our method. However,
solving a GEP involving two extremely ill-conditioned PSD matrices is both
expensive and numerically challenging, whereas fast, generic tools for solving
linear systems can be applied to our formulation, as we illustrate in the
following subsection. In addition, in many irregular sampling cases, the support
or domain of the data can be well-described by a union of convex regions. Using
prolate functions on such domains as right-hand side vectors $\bm{b}$ often
results in weights $\bal$ which are effectively equivalent to those computed
using the GPSS formulation. From this perspective, the approach we develop here
can be viewed as an alternative to Bronez's which is more amenable to
acceleration using modern computational techiniques.

As a final comment, we note that in the case where $\set{x_j}$ has been sampled
on a  grid and one selects $\Omega = (x_2 - x_1)^{-1}/2$ (the Nyquist
frequency), the GEP above simplifies to a standard eigenvalue problem as $\bm{B}
\propto \I$ (an observation also made by \cite{chave2019}). In such a setting,
matrix vector products with the sinc matrix $\bm{A}$ can be accelerated using
the \emph{fast sinc transform}~\cite{greengard2007}, and thus the GPSS weights
$\bm{\beta}$ which are solutions to $\bm{A} \bm{\beta} = \lambda \bm{\beta}$ can
be computed rapidly using a Krylov method.

\subsection{Numerical methods for weight computation} \label{sec:numerical}

Computing $\bm{\alpha}$ by solving the dense linear system $\bm{\Frt} \bal =
\bm{b}$ requires $\bO(n^3)$ operations if $\Omega = \bO(n)$. Fortunately,
several structural aspects of the problem can be exploited to accelerate the
computation of weights $\bal$. The first observation is that matrix-vector
products with $\bm{\Fr}$ can be computed in $\bO(n\log n)$ time using the NUFFT,
and so computational methods that require only matrix-vector products can be
accelerated.

One particularly simple case is for fixed $\Omega = \bO(1)$. In this setting,
the numerical rank of of $\Frt$ will stay roughly constant as $n$ grows.
Therefore, using rank-revealing~\cite{gu1996} or randomized~\cite{halko2011}
early-terminating factorization methods to obtain a low-rank approximation $\Frt
\approx \bm{U} \bm{V}^*$ with $\bm{U}, \bm{V} \in \mathbb{C}^{n \times r}$ will
require only $\bO(n)$ and $\bO(n\log n)$ operations respectively for bounded
rank $r$. Once such a factorization is obtained, standard pseudoinverse
identities can be used to obtain $\bal$ in $\bO(n)$ work.

In the more challenging case where $\Omega = \bO(n)$, the numerical rank of
$\Frt$ will be $\bO(n)$, and so early-terminating approaches would not improve
the runtime complexity of solving (\ref{eq:weights_linsys}). However, the NUFFT
still offers accelerated matrix-vector products in $\bO(n \log n)$ time, and
thus computation of $\bal$ can also be accelerated in this setting. In
particular, If one chooses $\set{\omg_j}_{j=1}^n$ to be Gauss-Legendre nodes on
$[-\Omega, \Omega]$ and defines the diagonal matrix $\bm{D}_{jj} = \gamma_j$ to
be the corresponding quadrature weights, we may transform
(\ref{eq:weights_linsys}) to $ \sqrt{\bm{D}} \Frt \bal = \sqrt{\bm{D}} \bm{b}$
which corresponds to normal equations given by 
$
  \Frt^* \bm{D} \Frt \bal = \Frt^* \bm{D} \bm{b}.
$
The key observation is that, by design,
\begin{equation} \label{eq:sinc_system}
  (\Frt^* \bm{D} \Frt)_{jk} \approx 2 \Omega \text{sinc}(2 \Omega (x_j - x_k)).
\end{equation}
Therefore, one can cheaply assemble an approximation $\bm{P} \approx \Frt^*
\bm{D} \Frt + \delta \I$ for $\delta$ near machine epsilon to this sinc matrix
which can be used as a preconditioner~\citep{murphy2000} in conjunction with an
iterative solver~\citep{fong2011} to achieve rapid convergence for good RHS
vectors $\bm{b}$ that can be accurately expressed as linear combinations of the
dominant eigenvectors of $\Frt^* \bm{D} \Frt$ (see \cite[Section 6.11]{saad2003}
for a rigorous analysis in the context of conjugate gradients). Using this
approach, the software companion to this work can compute weights satisfying
(\ref{eq:weights_linsys}) to effectively machine precision for hundreds of
thousands of fully irregular points $x_j$ in under a minute, as we demonstrate
in the next section.

Choosing an effective preconditioner for this problem is, in our experience,
dimension-dependent. In one dimension, a hierarchical
matrix~\cite{hackbusch2015} approximation to the matrix $\bm{P}_{jk} = 2 \Omega
\text{sinc}(2 \Omega (x_j - x_k)) + \delta \ind{j=k}$ can be assembled to high
accuracy because the sinc function is the product of two kernels: a $1/(x_j -
x_k)$ kernel, which is smooth away from the origin and thus leads to
off-diagonal blocks with low ranks~\cite{cheng1999fast, hackbusch2015}, and the
oscillatory kernel $\sin(2 \Omega (x_j - x_k))$. Note that any matrix
$\bm{M}_{jk} = \sin(2 \Omega (x_j - x_k))$ will be at most rank two because it
is the imaginary part of a Hermitian rank-one matrix: $\bm{M} = \text{Im}(\bm{u}
\bm{u}^*)$ where $\bm{u}_j = e^{-2 \pi i \Omega x_j}$. Therefore the rank of
each off-diagonal block of the $1/(x_j - x_k)$ kernel matrix will be at most
doubled by the oscillatory term, as
\begin{equation*} %\label{eq:}
  \underbrace{\bm{U} \bm{V}^T}_{\substack{\text{low rank block} \\ \text{of $\frac{1}{x_j - x_k}$}}}
  \circ \hspace{6pt}
  \Big(\underbrace{\bu_1 \bv_1^T + \bu_2 \bv_2^T}_{\substack{\text{rank two block} \\ \text{of $\sin\hspace{-2pt}\big(2\Omega(x_j - x_k)\big)$}}}\Big)
  =
  \bm{D}_{\bu_1} \bm{U} \bm{V}^T \bm{D}_{\bv_1} 
  +
  \bm{D}_{\bu_2} \bm{U} \bm{V}^T \bm{D}_{\bv_2}
\end{equation*}
by standard properties of Hadamard products, where $\bm{D}_{\bm{v}} :=
\text{diag}(v_1,\dots,v_n)$. As a result, the off-diagonal blocks of the sinc
matrix inherit the bounded ranks of the $1/(x_j - x_k)$ kernel, and a
hierarchical matrix approximation to $\bm{P}$ can thus be assembled and used to
solve linear systems in $\bO(n \log n)$ time in one dimension.

Unfortunately, the above rank observations for the oscillatory term do not apply
in higher dimensions, and empirically we find that preconditioning using
hierarchical matrices does not scale well in two dimensions. To motivate a more
general preconditioning strategy, we start by making the observation that for
\emph{any} kernel function $K(\bx - \bx') = \int_{\R^d} e^{2 \pi i
\bomg^T(\bx-\bx')} G_K(\bomg) \dif \bomg$ where $G_K$ is smooth and supported on
$[-\Omega, \Omega]^d$, one may instead select the $\bm{D}$ matrix
in~\eqref{eq:sinc_system} to have diagonal values $\bm{D}_{jj} = \gamma_j
G_K(\bomg_j)$. Then the matrix quadratic form in the normal equations will have
entries approximately given by $(\Frt^* \bm{D} \Frt)_{jk} \approx K(\bx_j -
\bx_k)$. A useful (but likely non-optimal) choice that we employ in this work is
to select $K$ to be a Gaussian function with bandwidth proportional to $\Omega$.
In this case, one may approximate $\bm{P}_{jk} = K(\bx_j - \bx_k)$ with a sparse
matrix by dropping near-zero values corresponding to the kernel at sufficiently
far distances, and use this sparse matrix as a preconditioner. Assessing the
runtime cost of factorizing and subsequently solving linear systems with a
sparse matrix is notoriously difficult because it depends on the sparsity
pattern. But for compact kernels in $2$D, the runtime cost will generally be
$\bO(n^{3/2})$~\cite{lipton1979}.

As a special case, we note that if locations $\set{x_j}$ are given on a gappy
lattice in any dimension, then $\Frt^* \bm{D} \Frt \approx 2 \Omega \I$. Thus in
this special case, despite severe ill-conditioning, iterative methods for the
transformed system will converge rapidly for well-chosen right-hand sides
$\bm{b}$. Finally, we note that there is a growing body of literature on the
problem of direct NUFFT inversion. In particular, methods such as the ones
described in~\cite{kircheis2019direct, kircheis2023fast, wilber2024} may also be
applicable or adaptable in some settings.

\begin{table}
  \centering
  \renewcommand{\arraystretch}{1.3}
  \begin{tabular}{c|c|c|c|c|c} 
    & \multicolumn{2}{c|}{\textbf{Compute} $\bm{b} = [G(\xi_j)]_{j=1}^n$} &
    \multicolumn{2}{c|}{\textbf{Solve} $\bm{\Frt} \bal = \bm{b}$} &
    \textbf{Evaluate} $\{S(\xi_k)\}_{k=1}^m$ \\ \hline 
    \multirow{2}{*}{\textbf{Dense}} & Prolate & Closed form &
    \multicolumn{2}{c|}{\multirow{2}{*}{$\bO(n^3)$}} &
    \multirow{2}{*}{$\bO(nm)$} \\ \cline{2-3} & $\bO(n^3)$ &
    \multirow{3}{*}{$\bO(n)$} & \multicolumn{2}{c|}{} & \\ \cline{1-2}
    \cline{4-6} \multirow{2}{*}{\textbf{Accelerated}} &
    \multirow{2}{*}{$\bO(n\log n)$} & & 1D & 2D & \multirow{2}{*}{$\bO(n + m\log
    m)$} \\ \cline{4-5} & & & $\bO(n\log n)$ & $\bO(n^{3/2})$ &
  \end{tabular}
  \caption{Complexities of the steps of our algorithm in one and two dimensions
  using na\"ive dense linear algebra as well as the accelerated methods proposed
  here.}
  \label{tab:complexity}
\end{table}

Having described all the subroutines needed to compute our spectral estimator,
we summarize their asymptotic costs in Table~\ref{tab:complexity}. These
subroutines include computing prolate functions $g$ and their Fourier transforms
$G$ for various domains in physical and spectral space (or simply evaluating a
closed form $G$ for e.g. the Kaiser window), solving the weight linear system
$\bm{\Frt} \bal = \bm{b}$ using a preconditioned iterative method, and
evaluating the estimator $\hat{S}$ at a set of desired frequencies using the
NUFFT. Note that all costs are quasilinear in the number of observations $n$ and
the number of estimated frequencies $m$ in one dimension, and sub-quadratic in
two dimensions. The accelerated cost of solving $\bm{\Frt} \bal = \bm{b}$
in Table~\ref{tab:complexity} assumes that the number of Krylov iterations is
constant in $n$. We see this behavior empirically (see Section~\ref{sec:demo}),
but a rigorous proof remains for future work.

\subsection{Controlling aliasing bias} \label{sec:analysis}

With the discussion of properly computing the weights $\bm{\alpha}$ complete, we
now turn to an analysis of the aliasing-based sources of bias in the estimator
$\hat{S}(\xi)$. For the duration of this section, we will assume the following:
\begin{itemize}
\item[(A$1$)] Measurement locations $\set{x_j}_{j=1}^n$ are i.i.d. samples from
a probability density function $p$.
\item[(A$2$)] The sampling density $p$ is symmetric about the origin and
supported on $[-a, a]$, so that its corresponding characteristic function
$\varphi(t) = \E_{X \sim p} e^{2 \pi i t X}$ is real-valued.
\item[(A$3$)] The weights $\bm{\alpha}$ have been computed to sufficient
accuracy that we may treat $H_{\alpha}(\omg) \approx G(\omg)$ as an equality for
$\abs{\omg} < \Omega$.
\item[(A$4$)] The maximum controlled frequency $\Omega$ for weights
$\bm{\alpha}$ is large enough that $\varphi(2 \Omega) \gg \varphi(\Omega)^2$,
and we may treat $\varphi(2 \Omega) - \varphi(\Omega)^2 \approx \varphi(2
\Omega)$ as an equality.
\item[(A$5$)] $\Omega$ may grow with $n$, but only at a rate such that there
exist weights $\bm{\alpha}$ with $\norm[1]{\bm{\alpha}} = \bO(1)$. Consequently,
the dependence of $\Omega$ and $\bal$ on $n$ will be suppressed.
\end{itemize}

We note that (A2) and (A4) are for notational simplicity and can be removed with
no impact on the results beyond additional notation. The assumption (A5) is
milder than it seems, and simply rules out edge cases in which $\Omega$
asymptotically converges to the Nyquist frequency in a way that induces the
weights $\bal$ to become oscillatory and grow in norm. As a reminder, by
construction one has that $\sum_{j=1}^n \alpha_j \approx G_0 := G(0)$, and so if
$\set{\alpha_j}_{j=1}^n$ is not overly oscillatory then $\norm[1]{\bal} \approx
G_0$ and this condition will naturally be satisfied.

Under these assumptions, we now study the two aliasing error-based bias terms
shown in (\ref{eq:bias_decomp}). Writing the estimator as $\hat{S}(\xi_j) =
|\eta_j|^2$ with $\bm{\eta} = \Frt \cdot \text{Diag}(\bal) \cdot \by$ and
letting $\hat{\bm{s}} = [\hat{S}(\xi_j)]_{j=1}^n$ denote the vector of
estimators we have that
\begin{equation*} %\label{eq:}
  \E \hat{\bm{s}} = \text{Diag}(\bm{M}) := \text{Diag}\Big(\Frt \text{Diag}(\bal)
  \bS \text{Diag}(\bal) \Frt^*\Big),
\end{equation*}
where $\by = [Y(x_j)]_{j=1}^n$ and $\bS = \E \by \by^T$ is its covariance
matrix. Therefore, by~\eqref{eq:bias_decomp} we have that
\begin{equation} \label{eq:epsj}
  \eps(\xi_j) = \bm{M}_{jj} - \int_{-\Omega + \xi_j}^{\Omega + \xi_j} |G(\omg - \xi_j)|^2
  S(\omg) \dif \omg
  =
  \int_{E_{\Omega, \xi_j}} |H_{\alpha}(\omg - \xi_j)|^2 S(\omg) \dif \omg,
\end{equation}
where $E_{\Omega, \xi_j} = \R \setminus [-\Omega + \xi_j, \Omega + \xi_j]$, is
precisely the aliasing-based bias in the estimator $\hat{S}(\xi_j)$. This error
is of size 
\begin{equation*} %\label{eq:}
  \E \eps(\xi_j)
  =
  \int_{E_{\Omega, \xi_j}} \E |H_{\alpha}(\omg - \xi)|^2 S(\omg) \dif \omg
  \approx
  \norm[2]{\bm{\alpha}}^2 \int_{E_{\Omega, \xi_j}} S(\omg) \dif \omg,
\end{equation*}
where the expectation is taken over the random locations $\set{x_j}_{j=1}^n$.
Thus we see that controlling the size of $\norm[2]{\bal}$ is critical to
reducing this aliasing bias. To obtain a more precise concentration-type bound
in terms of its relative size over the quantity to be estimated, $S(\xi)$, we
provide the following theorem whose proof can be found in
Appendix~\ref{app:proofs}.
\begin{theorem} \label{thm:aliasing} Let $S$ be a valid spectral density, let
  points $\set{x_j}_{j=1}^n \iid p$, and let $\bm{\alpha} \in \R^n$ be
  corresponding weights to resolve the window function $g$. Then under
  assumptions (A$1$)-(A$5$), 
  \begin{equation*} %\label{eq:}
    P\set{\eps(\xi) \geq \beta S(\xi)}
    \leq
     2 
     \exp\set{
     -
       \frac{
          \beta S(\xi)
       }{
          2 \norm[2]{\bal}^2 \int_{E_{\Omega, \xi}} S(\omg) \dif \omg
       }
     }.
  \end{equation*}
\end{theorem} 
This theorem gives a probabilistic control over the size of the aliasing bias
$\eps(\xi)$ relatively to $S(\xi)$. In particular, the event $\set{\eps(\xi)
\geq \beta S(\xi)}$ occuring would mean that one should \emph{not} expect the
estimator $\hat{S}(\xi)$ to achieve $-\log_{10}\beta$ correct digits in the
sense that $|\hat{S}(\xi) - S(\xi)|/S(\xi) < \beta$. It does not directly
correspond to the event of achieving $-\log_{10}\beta$ digits, but if $\eps(\xi)
> S(\xi)$, for example, then there is no chance of achieving even one correct
digit.  With this in mind, this bound is best interpreted as one that is
informative about disqualifying events in which a certain relative error is
achieved. A visual example of the phenomenon described by the above theorem is
given in Figure~\ref{fig:theorem_demo}. The key observation is that as $\xi$
increases, frequencies $\omg$ for which $S(\omg)$ is large --- typically in a
neighborhood of the origin --- are no longer included in the controlled
frequency band $[-\Omega + \xi, \Omega + \xi]$. Since $H_{\alpha}(\omg) \propto
\norm[2]{\bal}^2$ in the uncontrolled region $E_{\Omega,\xi}$, the bias
$\eps(\xi)$ can grow dramatically. In terms of relative error, this bias
can be even more significant, as the signal $S(\xi)$ typically shrinks as $\xi$
grows.

\begin{figure}[ht!]
  \centering
% GNUPLOT: LaTeX picture with Postscript
\begingroup
  \makeatletter
  \providecommand\color[2][]{%
    \GenericError{(gnuplot) \space\space\space\@spaces}{%
      Package color not loaded in conjunction with
      terminal option `colourtext'%
    }{See the gnuplot documentation for explanation.%
    }{Either use 'blacktext' in gnuplot or load the package
      color.sty in LaTeX.}%
    \renewcommand\color[2][]{}%
  }%
  \providecommand\includegraphics[2][]{%
    \GenericError{(gnuplot) \space\space\space\@spaces}{%
      Package graphicx or graphics not loaded%
    }{See the gnuplot documentation for explanation.%
    }{The gnuplot epslatex terminal needs graphicx.sty or graphics.sty.}%
    \renewcommand\includegraphics[2][]{}%
  }%
  \providecommand\rotatebox[2]{#2}%
  \@ifundefined{ifGPcolor}{%
    \newif\ifGPcolor
    \GPcolortrue
  }{}%
  \@ifundefined{ifGPblacktext}{%
    \newif\ifGPblacktext
    \GPblacktexttrue
  }{}%
  % define a \g@addto@macro without @ in the name:
  \let\gplgaddtomacro\g@addto@macro
  % define empty templates for all commands taking text:
  \gdef\gplbacktext{}%
  \gdef\gplfronttext{}%
  \makeatother
  \ifGPblacktext
    % no textcolor at all
    \def\colorrgb#1{}%
    \def\colorgray#1{}%
  \else
    % gray or color?
    \ifGPcolor
      \def\colorrgb#1{\color[rgb]{#1}}%
      \def\colorgray#1{\color[gray]{#1}}%
      \expandafter\def\csname LTw\endcsname{\color{white}}%
      \expandafter\def\csname LTb\endcsname{\color{black}}%
      \expandafter\def\csname LTa\endcsname{\color{black}}%
      \expandafter\def\csname LT0\endcsname{\color[rgb]{1,0,0}}%
      \expandafter\def\csname LT1\endcsname{\color[rgb]{0,1,0}}%
      \expandafter\def\csname LT2\endcsname{\color[rgb]{0,0,1}}%
      \expandafter\def\csname LT3\endcsname{\color[rgb]{1,0,1}}%
      \expandafter\def\csname LT4\endcsname{\color[rgb]{0,1,1}}%
      \expandafter\def\csname LT5\endcsname{\color[rgb]{1,1,0}}%
      \expandafter\def\csname LT6\endcsname{\color[rgb]{0,0,0}}%
      \expandafter\def\csname LT7\endcsname{\color[rgb]{1,0.3,0}}%
      \expandafter\def\csname LT8\endcsname{\color[rgb]{0.5,0.5,0.5}}%
    \else
      % gray
      \def\colorrgb#1{\color{black}}%
      \def\colorgray#1{\color[gray]{#1}}%
      \expandafter\def\csname LTw\endcsname{\color{white}}%
      \expandafter\def\csname LTb\endcsname{\color{black}}%
      \expandafter\def\csname LTa\endcsname{\color{black}}%
      \expandafter\def\csname LT0\endcsname{\color{black}}%
      \expandafter\def\csname LT1\endcsname{\color{black}}%
      \expandafter\def\csname LT2\endcsname{\color{black}}%
      \expandafter\def\csname LT3\endcsname{\color{black}}%
      \expandafter\def\csname LT4\endcsname{\color{black}}%
      \expandafter\def\csname LT5\endcsname{\color{black}}%
      \expandafter\def\csname LT6\endcsname{\color{black}}%
      \expandafter\def\csname LT7\endcsname{\color{black}}%
      \expandafter\def\csname LT8\endcsname{\color{black}}%
    \fi
  \fi
    \setlength{\unitlength}{0.0500bp}%
    \ifx\gptboxheight\undefined%
      \newlength{\gptboxheight}%
      \newlength{\gptboxwidth}%
      \newsavebox{\gptboxtext}%
    \fi%
    \setlength{\fboxrule}{0.5pt}%
    \setlength{\fboxsep}{1pt}%
    \definecolor{tbcol}{rgb}{1,1,1}%
\begin{picture}(8500.00,3400.00)%
    \gplgaddtomacro\gplbacktext{%
      \csname LTb\endcsname%%
      \put(323,676){\makebox(0,0)[r]{\strut{}\footnotesize $10^{-6}$}}%
      \csname LTb\endcsname%%
      \put(323,1149){\makebox(0,0)[r]{\strut{}\footnotesize $10^{-5}$}}%
      \csname LTb\endcsname%%
      \put(323,1622){\makebox(0,0)[r]{\strut{}\footnotesize $10^{-4}$}}%
      \csname LTb\endcsname%%
      \put(323,2095){\makebox(0,0)[r]{\strut{}\footnotesize $10^{-3}$}}%
      \csname LTb\endcsname%%
      \put(323,2568){\makebox(0,0)[r]{\strut{}\footnotesize $10^{-2}$}}%
      \csname LTb\endcsname%%
      \put(323,3041){\makebox(0,0)[r]{\strut{}\footnotesize $10^{-1}$}}%
      \csname LTb\endcsname%%
      \put(424,436){\makebox(0,0){\strut{}0}}%
      \csname LTb\endcsname%%
      \put(1560,436){\makebox(0,0){\strut{}$\Omega/2$}}%
      \csname LTb\endcsname%%
      \put(2697,436){\makebox(0,0){\strut{}$\Omega$}}%
    }%
    \gplgaddtomacro\gplfronttext{%
      \csname LTb\endcsname%%
      \put(2525,2826){\makebox(0,0)[r]{\strut{}\footnotesize $\E \hat{S}(\omg)$}}%
      \csname LTb\endcsname%%
      \put(2525,2586){\makebox(0,0)[r]{\strut{}\footnotesize $S(\omg)$}}%
      \csname LTb\endcsname%%
      \put(2525,2346){\makebox(0,0)[r]{\strut{}\footnotesize $(S * |G|^2)(\omg)$}}%
      \csname LTb\endcsname%%
      \put(2525,2107){\makebox(0,0)[r]{\strut{}\footnotesize $\eps(\omg)$}}%
      \csname LTb\endcsname%%
      \put(2119,76){\makebox(0,0){\strut{}\small $\omega$}}%
    }%
    \gplgaddtomacro\gplbacktext{%
      \csname LTb\endcsname%%
      \put(4563,676){\makebox(0,0)[r]{\strut{}\footnotesize $10^{-25}$}}%
      \csname LTb\endcsname%%
      \put(4563,1149){\makebox(0,0)[r]{\strut{}\footnotesize $10^{-20}$}}%
      \csname LTb\endcsname%%
      \put(4563,1622){\makebox(0,0)[r]{\strut{}\footnotesize $10^{-15}$}}%
      \csname LTb\endcsname%%
      \put(4563,2095){\makebox(0,0)[r]{\strut{}\footnotesize $10^{-10}$}}%
      \csname LTb\endcsname%%
      \put(4563,2568){\makebox(0,0)[r]{\strut{}\footnotesize $10^{-5}$}}%
      \csname LTb\endcsname%%
      \put(4563,3041){\makebox(0,0)[r]{\strut{}\footnotesize $10^{0}$}}%
      \csname LTb\endcsname%%
      \put(5004,436){\makebox(0,0){\strut{}$-\Omega + \xi$}}%
      \csname LTb\endcsname%%
      \put(6360,436){\makebox(0,0){\strut{}$\xi$}}%
      \csname LTb\endcsname%%
      \put(7715,436){\makebox(0,0){\strut{}$\Omega + \xi$}}%
    }%
    \gplgaddtomacro\gplfronttext{%
      \csname LTb\endcsname%%
      \put(7305,1610){\makebox(0,0)[r]{\strut{}\footnotesize (n=5k)}}%
      \csname LTb\endcsname%%
      \put(7305,1371){\makebox(0,0)[r]{\strut{}\footnotesize (n=100k)}}%
      \csname LTb\endcsname%%
      \put(6359,76){\makebox(0,0){\strut{}\small $\omega$}}%
    }%
    \gplbacktext
    \put(0,0){\includegraphics[width={425.00bp},height={170.00bp}]{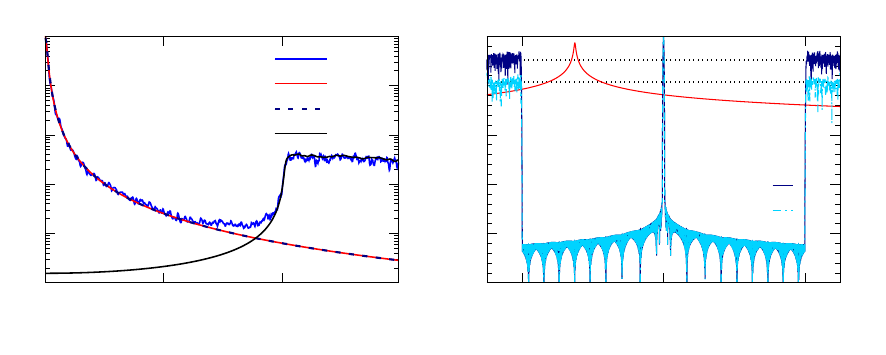}}%
    \gplfronttext
  \end{picture}%
\endgroup
  \caption{ A visual demonstration of the phenomenon described in Theorem
  \ref{thm:aliasing}. Left panel: the expected value of the estimator
  $\hat{S}(\omg)$ at various frequencies compared with the true $S(\omg)$ and
  the aliasing bias $\eps(\omg)$. Right panel: $|H_{\alpha}(\omg - \xi)|^2$ for
  two values of $n$ but fixed $\Omega$ visualized against $S(\omg)$, where the
  integral of the product of those two functions is $\E \hat{S}(\omg)$. The two
  faint dotted lines give $\norm[2]{\bal}^2$ for each data size. }
  \label{fig:theorem_demo}
\end{figure}

Several consequences and interpretive conclusions follow from this theorem.
First, we see that if $\norm[2]{\bm{\alpha}}$ were to blow up, the bound on
aliasing errors would lose all power. As discussed in Sections~\ref{sec:weights}
and~\ref{sec:window}, this blowup happens precisely when one either tries to
resolve too high of a maximum controlled frequency $\Omega$ or selects a poor
window function $g$. As Equation~\ref{eq:bias_decomp} and this theorem make
clear, this aliasing error will pollute \emph{every} estimate $\hat{S}(\xi)$,
even if $\xi \ll \Omega$. So before computing any estimator, one must confirm
that the weights $\bm{\alpha}$ are not only a high-accuracy solution to
(\ref{eq:weights_linsys}), but also are small in norm. Finally, we observe that
the sampling scheme $p$ (or equivalently $\varphi$) does not directly impact
this tail aliasing error in the sense that $\varphi$ is asymptotically
irrelevant to the distribution of $|H_{\alpha}|^2$. It is important to keep in
mind, however, that the distribution of sampling locations plays a crucial role
in the size of $\norm[2]{\bal}$ in practice, and thus should be considered when
selecting a window function $g$.

We provide the following result regarding the convergence rate for spectral
densities with algebraic decay, such as the Mat\'ern process. 
\begin{corollary} \label{cor:rate} Let $Y(x)$ be a process whose spectral
  density $S$ has algebraic tails ${S(\omg) \propto \abs{\omg}^{-2 \nu - 1}}$
  for $|\omg| > \Omega_0$. Consider sampling $Y$ at locations $x_1,\dots,x_n
  \iid p$ satisfying (A$1$)-(A$2$). Let $\Omega_n$ be a sequence of maximum
  controlled frequencies satisfying (A$3$)-(A$5$). Then for frequencies $\xi_n <
  \Omega_n$ and $n$ sufficiently large, 
  \begin{equation*} %\label{eq:} 
    \E \eps(\xi_n) \leq \frac{C (\Omega_n - \xi_n)^{-2 \nu}}{\nu},
  \end{equation*}
  where $C$ is a constant that does not depend on the particular locations
  $\set{x_j}$. 
  \end{corollary}
The error between the estimated value $\hat{S}$ and the convolved spectrum $(G *
S)$ is made up of the error between $G$ and $H_\alpha$ within the controlled
band $[-\Omega, \Omega]$, and the aliasing error studied above. In practice this
first source of error is negligible, and thus Corollary~\ref{cor:rate} can be
used to obtain the rate at which $\hat{S}(\xi_n) \to_n (G * S)(\xi_n)$ for
sequences $\set{\xi_n}_n$ such that $\Omega_n - \xi_n$ has a (possibly infinite)
limit.

A consequence of this result is that for processes with slow spectral decay, one
can require significantly higher amounts of data to reduce $\E \eps(\xi_n)$ by a
meaningful factor. As an example, consider a Mat\'ern process with $\nu =
\frac{1}{2}$, taking $\xi_n = \sqrt{n}$ and $\Omega_n = \frac{1}{5} n^{0.9}$. If
we start with a baseline data size of $n=1{,}000$, then $\E \eps(\xi_n)$ is
reduced by a factor of $10$ at data size $n' \approx 10{,}000$. For $\nu =
\frac{5}{2}$, on the other hand, the same reduction is achieved at $n' \approx
1{,}550$. 

As a final observation on controlling $\eps(\xi_n)$, we note that if $\Omega_n$
is chosen \emph{less} ambitiously than the maximum possible under which
$\bal(n)$ remain uniformly bounded, then one can achieve $\norm[2]{\bal(n)}
\rightarrow_n 0$, which also leads to a reduction in aliasing error. Therefore,
particularly for processes with very slow spectral decay, more conservative
choices of $\Omega$ which reduce the size of $\norm[2]{\bal}$ may yield superior
performance. While we leave a detailed analysis of such tradeoffs to future work
and specific applications, it is worth keeping in mind that a larger $\Omega$ is
not the only way to reduce super-Nyquist bias.

\section{Numerical demonstrations} \label{sec:demo}

In the following section we provide a variety of numerical experiments
illuminating various aspects of our spectral estimator, including aliasing bias
bounds and window reconstruction errors. In addition, we demonstrate the
performance of our method in estimating spectral densities of one- and
two-dimensional processes.

\subsection{Runtime cost of weight computation} \label{sec:runtime}

Figures \ref{fig:runtime} and \ref{fig:runtime_2d} show the runtime cost of
computing weights in two cases: random locations $\set{x_j}_{j=1}^n \iid
\text{Unif}([0,1])$, and $\set{x_j}_{j=1}^n$ derived from a grid with a single
large gap in one and two dimensions respectively. The biggest effective
difference between these two computations is that for the uniformly random
points, a preconditioner must be assembled and factorized in order to accelerate
the Krylov solver. In the gappy gridded case, as mentioned above, the matrix
$\Frt^* \bm{D} \Frt$ in the normal equations reduces approximately to a scaled
identity matrix --- and so even with no preconditioner and a severely
ill-conditioned linear system, the iterative solver still converges rapidly. In
all cases, the number of iterations required to solve the system is effectively
constant in the number of observations $n$, and so the dominant runtime cost in
the case of fully irregular data is preconditioner construction, and in the
gridded case is the NUFFT computation. Building the preconditioner is admittedly
expensive relative to computing all of the necessary NUFFTs (costing roughly a
factor of five times more in the current software implementation and data sizes
of Figure \ref{fig:runtime}).  While the asymptotic complexity of forming and
solving linear systems with the sparse preconditioner in two dimensions is
$\bO(n^{3/2})$ in theory, evidently that scaling cost can be slow to dominate,
as the agreement with $\bO(n \log n)$ scaling is quite good even up to
$n=2^{18}$ points. 

\begin{figure}[!ht]
  \centering
% GNUPLOT: LaTeX picture with Postscript
\begingroup
  \makeatletter
  \providecommand\color[2][]{%
    \GenericError{(gnuplot) \space\space\space\@spaces}{%
      Package color not loaded in conjunction with
      terminal option `colourtext'%
    }{See the gnuplot documentation for explanation.%
    }{Either use 'blacktext' in gnuplot or load the package
      color.sty in LaTeX.}%
    \renewcommand\color[2][]{}%
  }%
  \providecommand\includegraphics[2][]{%
    \GenericError{(gnuplot) \space\space\space\@spaces}{%
      Package graphicx or graphics not loaded%
    }{See the gnuplot documentation for explanation.%
    }{The gnuplot epslatex terminal needs graphicx.sty or graphics.sty.}%
    \renewcommand\includegraphics[2][]{}%
  }%
  \providecommand\rotatebox[2]{#2}%
  \@ifundefined{ifGPcolor}{%
    \newif\ifGPcolor
    \GPcolortrue
  }{}%
  \@ifundefined{ifGPblacktext}{%
    \newif\ifGPblacktext
    \GPblacktexttrue
  }{}%
  % define a \g@addto@macro without @ in the name:
  \let\gplgaddtomacro\g@addto@macro
  % define empty templates for all commands taking text:
  \gdef\gplbacktext{}%
  \gdef\gplfronttext{}%
  \makeatother
  \ifGPblacktext
    % no textcolor at all
    \def\colorrgb#1{}%
    \def\colorgray#1{}%
  \else
    % gray or color?
    \ifGPcolor
      \def\colorrgb#1{\color[rgb]{#1}}%
      \def\colorgray#1{\color[gray]{#1}}%
      \expandafter\def\csname LTw\endcsname{\color{white}}%
      \expandafter\def\csname LTb\endcsname{\color{black}}%
      \expandafter\def\csname LTa\endcsname{\color{black}}%
      \expandafter\def\csname LT0\endcsname{\color[rgb]{1,0,0}}%
      \expandafter\def\csname LT1\endcsname{\color[rgb]{0,1,0}}%
      \expandafter\def\csname LT2\endcsname{\color[rgb]{0,0,1}}%
      \expandafter\def\csname LT3\endcsname{\color[rgb]{1,0,1}}%
      \expandafter\def\csname LT4\endcsname{\color[rgb]{0,1,1}}%
      \expandafter\def\csname LT5\endcsname{\color[rgb]{1,1,0}}%
      \expandafter\def\csname LT6\endcsname{\color[rgb]{0,0,0}}%
      \expandafter\def\csname LT7\endcsname{\color[rgb]{1,0.3,0}}%
      \expandafter\def\csname LT8\endcsname{\color[rgb]{0.5,0.5,0.5}}%
    \else
      % gray
      \def\colorrgb#1{\color{black}}%
      \def\colorgray#1{\color[gray]{#1}}%
      \expandafter\def\csname LTw\endcsname{\color{white}}%
      \expandafter\def\csname LTb\endcsname{\color{black}}%
      \expandafter\def\csname LTa\endcsname{\color{black}}%
      \expandafter\def\csname LT0\endcsname{\color{black}}%
      \expandafter\def\csname LT1\endcsname{\color{black}}%
      \expandafter\def\csname LT2\endcsname{\color{black}}%
      \expandafter\def\csname LT3\endcsname{\color{black}}%
      \expandafter\def\csname LT4\endcsname{\color{black}}%
      \expandafter\def\csname LT5\endcsname{\color{black}}%
      \expandafter\def\csname LT6\endcsname{\color{black}}%
      \expandafter\def\csname LT7\endcsname{\color{black}}%
      \expandafter\def\csname LT8\endcsname{\color{black}}%
    \fi
  \fi
    \setlength{\unitlength}{0.0500bp}%
    \ifx\gptboxheight\undefined%
      \newlength{\gptboxheight}%
      \newlength{\gptboxwidth}%
      \newsavebox{\gptboxtext}%
    \fi%
    \setlength{\fboxrule}{0.5pt}%
    \setlength{\fboxsep}{1pt}%
    \definecolor{tbcol}{rgb}{1,1,1}%
\begin{picture}(9060.00,3400.00)%
    \gplgaddtomacro\gplbacktext{%
      \csname LTb\endcsname%%
      \put(1255,992){\makebox(0,0)[r]{\strut{}\footnotesize $10^{-1}$}}%
      \csname LTb\endcsname%%
      \put(1255,1638){\makebox(0,0)[r]{\strut{}\footnotesize $10^{0}$}}%
      \csname LTb\endcsname%%
      \put(1255,2285){\makebox(0,0)[r]{\strut{}\footnotesize $10^{1}$}}%
      \csname LTb\endcsname%%
      \put(1356,436){\makebox(0,0){\strut{}\footnotesize $2^{10}$}}%
      \csname LTb\endcsname%%
      \put(1723,436){\makebox(0,0){\strut{}\footnotesize $2^{11}$}}%
      \csname LTb\endcsname%%
      \put(2090,436){\makebox(0,0){\strut{}\footnotesize $2^{12}$}}%
      \csname LTb\endcsname%%
      \put(2457,436){\makebox(0,0){\strut{}\footnotesize $2^{13}$}}%
      \csname LTb\endcsname%%
      \put(2825,436){\makebox(0,0){\strut{}\footnotesize $2^{14}$}}%
      \csname LTb\endcsname%%
      \put(3192,436){\makebox(0,0){\strut{}\footnotesize $2^{15}$}}%
      \csname LTb\endcsname%%
      \put(3559,436){\makebox(0,0){\strut{}\footnotesize $2^{16}$}}%
      \csname LTb\endcsname%%
      \put(3926,436){\makebox(0,0){\strut{}\footnotesize $2^{17}$}}%
      \csname LTb\endcsname%%
      \put(4293,436){\makebox(0,0){\strut{}\footnotesize $2^{18}$}}%
    }%
    \gplgaddtomacro\gplfronttext{%
      \csname LTb\endcsname%%
      \put(3745,1131){\makebox(0,0)[r]{\strut{}\footnotesize runtime}}%
      \csname LTb\endcsname%%
      \put(3745,891){\makebox(0,0)[r]{\strut{}\footnotesize $\bO(n \log n)$}}%
      \csname LTb\endcsname%%
      \put(455,1689){\rotatebox{-270.00}{\makebox(0,0){\strut{}\small runtime (s)}}}%
      \csname LTb\endcsname%%
      \put(2824,76){\makebox(0,0){\strut{}\small $n$}}%
      \csname LTb\endcsname%%
      \put(2824,3063){\makebox(0,0){\strut{}\small $x_j$ uniform i.i.d.}}%
    }%
    \gplgaddtomacro\gplbacktext{%
      \csname LTb\endcsname%%
      \put(5097,1089){\makebox(0,0)[r]{\strut{}\footnotesize $10^{-2}$}}%
      \csname LTb\endcsname%%
      \put(5097,1850){\makebox(0,0)[r]{\strut{}\footnotesize $10^{-1}$}}%
      \csname LTb\endcsname%%
      \put(5097,2612){\makebox(0,0)[r]{\strut{}\footnotesize $10^{0}$}}%
      \csname LTb\endcsname%%
      \put(5198,436){\makebox(0,0){\strut{}\footnotesize $2^{10}$}}%
      \csname LTb\endcsname%%
      \put(5565,436){\makebox(0,0){\strut{}\footnotesize $2^{11}$}}%
      \csname LTb\endcsname%%
      \put(5932,436){\makebox(0,0){\strut{}\footnotesize $2^{12}$}}%
      \csname LTb\endcsname%%
      \put(6299,436){\makebox(0,0){\strut{}\footnotesize $2^{13}$}}%
      \csname LTb\endcsname%%
      \put(6667,436){\makebox(0,0){\strut{}\footnotesize $2^{14}$}}%
      \csname LTb\endcsname%%
      \put(7034,436){\makebox(0,0){\strut{}\footnotesize $2^{15}$}}%
      \csname LTb\endcsname%%
      \put(7401,436){\makebox(0,0){\strut{}\footnotesize $2^{16}$}}%
      \csname LTb\endcsname%%
      \put(7768,436){\makebox(0,0){\strut{}\footnotesize $2^{17}$}}%
      \csname LTb\endcsname%%
      \put(8135,436){\makebox(0,0){\strut{}\footnotesize $2^{18}$}}%
    }%
    \gplgaddtomacro\gplfronttext{%
      \csname LTb\endcsname%%
      \put(7587,1131){\makebox(0,0)[r]{\strut{}\footnotesize runtime}}%
      \csname LTb\endcsname%%
      \put(7587,891){\makebox(0,0)[r]{\strut{}\footnotesize $\bO(n \log n)$}}%
      \csname LTb\endcsname%%
      \put(6666,76){\makebox(0,0){\strut{}\small $n$}}%
      \csname LTb\endcsname%%
      \put(6666,3063){\makebox(0,0){\strut{}\small $x_j$ gappy grid}}%
    }%
    \gplbacktext
    \put(0,0){\includegraphics[width={453.00bp},height={170.00bp}]{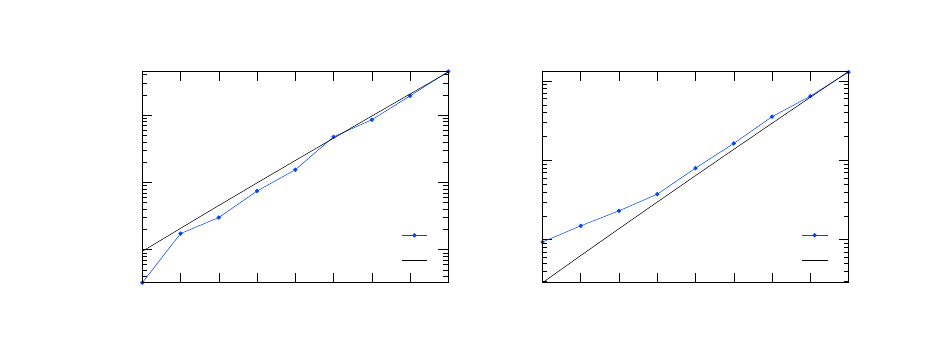}}%
    \gplfronttext
  \end{picture}%
\endgroup
  \caption{Runtime measurements for weight computations in two sampling regimes:
  uniformly random points and gappy gridded points. For random points, a
  hierarchical matrix preconditioner is used to accelerate the Krylov solver. In
  the gappy gridded case, no preconditioner is used.}
  \label{fig:runtime}
\end{figure}

\begin{figure}[!ht]
  \centering
% GNUPLOT: LaTeX picture with Postscript
\begingroup
  \makeatletter
  \providecommand\color[2][]{%
    \GenericError{(gnuplot) \space\space\space\@spaces}{%
      Package color not loaded in conjunction with
      terminal option `colourtext'%
    }{See the gnuplot documentation for explanation.%
    }{Either use 'blacktext' in gnuplot or load the package
      color.sty in LaTeX.}%
    \renewcommand\color[2][]{}%
  }%
  \providecommand\includegraphics[2][]{%
    \GenericError{(gnuplot) \space\space\space\@spaces}{%
      Package graphicx or graphics not loaded%
    }{See the gnuplot documentation for explanation.%
    }{The gnuplot epslatex terminal needs graphicx.sty or graphics.sty.}%
    \renewcommand\includegraphics[2][]{}%
  }%
  \providecommand\rotatebox[2]{#2}%
  \@ifundefined{ifGPcolor}{%
    \newif\ifGPcolor
    \GPcolortrue
  }{}%
  \@ifundefined{ifGPblacktext}{%
    \newif\ifGPblacktext
    \GPblacktexttrue
  }{}%
  % define a \g@addto@macro without @ in the name:
  \let\gplgaddtomacro\g@addto@macro
  % define empty templates for all commands taking text:
  \gdef\gplbacktext{}%
  \gdef\gplfronttext{}%
  \makeatother
  \ifGPblacktext
    % no textcolor at all
    \def\colorrgb#1{}%
    \def\colorgray#1{}%
  \else
    % gray or color?
    \ifGPcolor
      \def\colorrgb#1{\color[rgb]{#1}}%
      \def\colorgray#1{\color[gray]{#1}}%
      \expandafter\def\csname LTw\endcsname{\color{white}}%
      \expandafter\def\csname LTb\endcsname{\color{black}}%
      \expandafter\def\csname LTa\endcsname{\color{black}}%
      \expandafter\def\csname LT0\endcsname{\color[rgb]{1,0,0}}%
      \expandafter\def\csname LT1\endcsname{\color[rgb]{0,1,0}}%
      \expandafter\def\csname LT2\endcsname{\color[rgb]{0,0,1}}%
      \expandafter\def\csname LT3\endcsname{\color[rgb]{1,0,1}}%
      \expandafter\def\csname LT4\endcsname{\color[rgb]{0,1,1}}%
      \expandafter\def\csname LT5\endcsname{\color[rgb]{1,1,0}}%
      \expandafter\def\csname LT6\endcsname{\color[rgb]{0,0,0}}%
      \expandafter\def\csname LT7\endcsname{\color[rgb]{1,0.3,0}}%
      \expandafter\def\csname LT8\endcsname{\color[rgb]{0.5,0.5,0.5}}%
    \else
      % gray
      \def\colorrgb#1{\color{black}}%
      \def\colorgray#1{\color[gray]{#1}}%
      \expandafter\def\csname LTw\endcsname{\color{white}}%
      \expandafter\def\csname LTb\endcsname{\color{black}}%
      \expandafter\def\csname LTa\endcsname{\color{black}}%
      \expandafter\def\csname LT0\endcsname{\color{black}}%
      \expandafter\def\csname LT1\endcsname{\color{black}}%
      \expandafter\def\csname LT2\endcsname{\color{black}}%
      \expandafter\def\csname LT3\endcsname{\color{black}}%
      \expandafter\def\csname LT4\endcsname{\color{black}}%
      \expandafter\def\csname LT5\endcsname{\color{black}}%
      \expandafter\def\csname LT6\endcsname{\color{black}}%
      \expandafter\def\csname LT7\endcsname{\color{black}}%
      \expandafter\def\csname LT8\endcsname{\color{black}}%
    \fi
  \fi
    \setlength{\unitlength}{0.0500bp}%
    \ifx\gptboxheight\undefined%
      \newlength{\gptboxheight}%
      \newlength{\gptboxwidth}%
      \newsavebox{\gptboxtext}%
    \fi%
    \setlength{\fboxrule}{0.5pt}%
    \setlength{\fboxsep}{1pt}%
    \definecolor{tbcol}{rgb}{1,1,1}%
\begin{picture}(9060.00,3400.00)%
    \gplgaddtomacro\gplbacktext{%
      \csname LTb\endcsname%%
      \put(1255,1712){\makebox(0,0)[r]{\strut{}$10^{0}$}}%
      \csname LTb\endcsname%%
      \put(1255,2273){\makebox(0,0)[r]{\strut{}$10^{1}$}}%
      \csname LTb\endcsname%%
      \put(1356,436){\makebox(0,0){\strut{}\footnotesize $2^{10}$}}%
      \csname LTb\endcsname%%
      \put(1723,436){\makebox(0,0){\strut{}\footnotesize $2^{11}$}}%
      \csname LTb\endcsname%%
      \put(2090,436){\makebox(0,0){\strut{}\footnotesize $2^{12}$}}%
      \csname LTb\endcsname%%
      \put(2457,436){\makebox(0,0){\strut{}\footnotesize $2^{13}$}}%
      \csname LTb\endcsname%%
      \put(2825,436){\makebox(0,0){\strut{}\footnotesize $2^{14}$}}%
      \csname LTb\endcsname%%
      \put(3192,436){\makebox(0,0){\strut{}\footnotesize $2^{15}$}}%
      \csname LTb\endcsname%%
      \put(3559,436){\makebox(0,0){\strut{}\footnotesize $2^{16}$}}%
      \csname LTb\endcsname%%
      \put(3926,436){\makebox(0,0){\strut{}\footnotesize $2^{17}$}}%
      \csname LTb\endcsname%%
      \put(4293,436){\makebox(0,0){\strut{}\footnotesize $2^{18}$}}%
    }%
    \gplgaddtomacro\gplfronttext{%
      \csname LTb\endcsname%%
      \put(3745,1371){\makebox(0,0)[r]{\strut{}\footnotesize runtime}}%
      \csname LTb\endcsname%%
      \put(3745,1131){\makebox(0,0)[r]{\strut{}\footnotesize $\bO(n \log n)$}}%
      \csname LTb\endcsname%%
      \put(3745,891){\makebox(0,0)[r]{\strut{}\footnotesize $\bO(n^{3/2})$}}%
      \csname LTb\endcsname%%
      \put(757,1689){\rotatebox{-270.00}{\makebox(0,0){\strut{}\small runtime (s)}}}%
      \csname LTb\endcsname%%
      \put(2824,76){\makebox(0,0){\strut{}\small $n$}}%
      \csname LTb\endcsname%%
      \put(2824,3063){\makebox(0,0){\strut{}\small $\bx_j$ uniform i.i.d.}}%
    }%
    \gplgaddtomacro\gplbacktext{%
      \csname LTb\endcsname%%
      \put(5097,1121){\makebox(0,0)[r]{\strut{}\footnotesize $10^{-1}$}}%
      \csname LTb\endcsname%%
      \put(5097,1883){\makebox(0,0)[r]{\strut{}\footnotesize $10^{0}$}}%
      \csname LTb\endcsname%%
      \put(5097,2644){\makebox(0,0)[r]{\strut{}\footnotesize $10^{1}$}}%
      \csname LTb\endcsname%%
      \put(5198,436){\makebox(0,0){\strut{}\footnotesize $2^{10}$}}%
      \csname LTb\endcsname%%
      \put(5565,436){\makebox(0,0){\strut{}\footnotesize $2^{11}$}}%
      \csname LTb\endcsname%%
      \put(5932,436){\makebox(0,0){\strut{}\footnotesize $2^{12}$}}%
      \csname LTb\endcsname%%
      \put(6299,436){\makebox(0,0){\strut{}\footnotesize $2^{13}$}}%
      \csname LTb\endcsname%%
      \put(6667,436){\makebox(0,0){\strut{}\footnotesize $2^{14}$}}%
      \csname LTb\endcsname%%
      \put(7034,436){\makebox(0,0){\strut{}\footnotesize $2^{15}$}}%
      \csname LTb\endcsname%%
      \put(7401,436){\makebox(0,0){\strut{}\footnotesize $2^{16}$}}%
      \csname LTb\endcsname%%
      \put(7768,436){\makebox(0,0){\strut{}\footnotesize $2^{17}$}}%
      \csname LTb\endcsname%%
      \put(8135,436){\makebox(0,0){\strut{}\footnotesize $2^{18}$}}%
    }%
    \gplgaddtomacro\gplfronttext{%
      \csname LTb\endcsname%%
      \put(7587,1131){\makebox(0,0)[r]{\strut{}\footnotesize runtime}}%
      \csname LTb\endcsname%%
      \put(7587,891){\makebox(0,0)[r]{\strut{}\footnotesize $\bO(n \log n)$}}%
      \csname LTb\endcsname%%
      \put(6666,76){\makebox(0,0){\strut{}\small $n$}}%
      \csname LTb\endcsname%%
      \put(6666,3063){\makebox(0,0){\strut{}\small $\bx_j$ gappy grid}}%
    }%
    \gplbacktext
    \put(0,0){\includegraphics[width={453.00bp},height={170.00bp}]{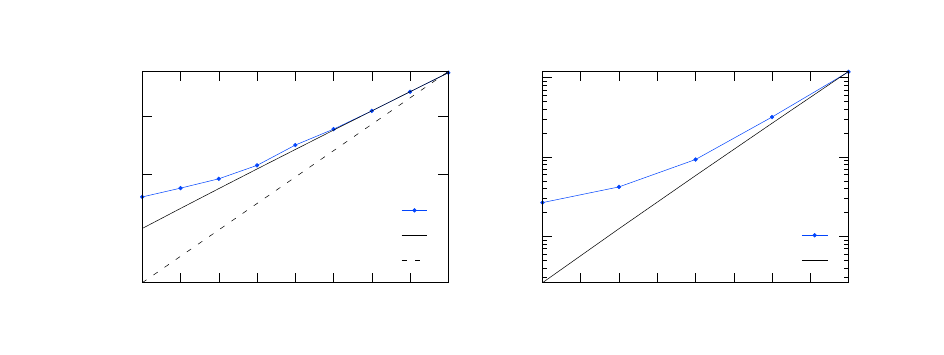}}%
    \gplfronttext
  \end{picture}%
\endgroup
  \caption{An analog of Figure \ref{fig:runtime} but in two dimensions. In the
  uniformly random case, the sparse Gaussian kernel preconditioner described in
  Section \ref{sec:numerical} is used.}
  \label{fig:runtime_2d}
\end{figure}

\subsection{Artifacts of interpolating to a regular grid}
\label{sec:interp_error}

A common practice in the spectral analysis of irregularly sampled data is to
interpolate the data to a regular grid so that standard equispaced spectral
estimators can be used. Certainly with sufficient care and in favorable settings
this \textit{``regridding''} can work very well: kernel interpolation using a
kernel function whose Fourier transform has tails which match the tails of the
process' true spectral density, for example, can result in minimal regridding
artifacts even at high frequencies. Given that high frequency information is one
of the few things that can be estimated well in the fixed-domain asymptotic
sampling regime \cite{stein1999}, one would expect that fitting e.g. a Mat\'ern
model (crucially including the smoothness parameter) to the data and
interpolating would perform well (see~\cite{babu2010} for more discussion, and
\cite{elipot2016,early2020} for a discussion and exploration of the impact of
different regridding methods).  However, estimating the necessary kernel
parameters can be expensive, and model misspecification can still result in
significant errors. The fundamental limitation of regridding estimators is that
the Fourier transform of the regridded data is simply a linear combination of
the Fourier transforms of the shifted kernel function, and thus spectral
estimates often exhibit the spectral characteristics of the interpolating kernel
rather than the true process.  When estimating the spectral density at
sufficiently low frequencies from sufficiently dense data, regridding is
unlikely to produce dramatic artifacts.  At higher frequencies compared to the
data sampling rate and basis function density, this issue can produce
meaningfully incorrect results.

\begin{figure}[!ht]
  \centering
% GNUPLOT: LaTeX picture with Postscript
\begingroup
  \makeatletter
  \providecommand\color[2][]{%
    \GenericError{(gnuplot) \space\space\space\@spaces}{%
      Package color not loaded in conjunction with
      terminal option `colourtext'%
    }{See the gnuplot documentation for explanation.%
    }{Either use 'blacktext' in gnuplot or load the package
      color.sty in LaTeX.}%
    \renewcommand\color[2][]{}%
  }%
  \providecommand\includegraphics[2][]{%
    \GenericError{(gnuplot) \space\space\space\@spaces}{%
      Package graphicx or graphics not loaded%
    }{See the gnuplot documentation for explanation.%
    }{The gnuplot epslatex terminal needs graphicx.sty or graphics.sty.}%
    \renewcommand\includegraphics[2][]{}%
  }%
  \providecommand\rotatebox[2]{#2}%
  \@ifundefined{ifGPcolor}{%
    \newif\ifGPcolor
    \GPcolortrue
  }{}%
  \@ifundefined{ifGPblacktext}{%
    \newif\ifGPblacktext
    \GPblacktexttrue
  }{}%
  % define a \g@addto@macro without @ in the name:
  \let\gplgaddtomacro\g@addto@macro
  % define empty templates for all commands taking text:
  \gdef\gplbacktext{}%
  \gdef\gplfronttext{}%
  \makeatother
  \ifGPblacktext
    % no textcolor at all
    \def\colorrgb#1{}%
    \def\colorgray#1{}%
  \else
    % gray or color?
    \ifGPcolor
      \def\colorrgb#1{\color[rgb]{#1}}%
      \def\colorgray#1{\color[gray]{#1}}%
      \expandafter\def\csname LTw\endcsname{\color{white}}%
      \expandafter\def\csname LTb\endcsname{\color{black}}%
      \expandafter\def\csname LTa\endcsname{\color{black}}%
      \expandafter\def\csname LT0\endcsname{\color[rgb]{1,0,0}}%
      \expandafter\def\csname LT1\endcsname{\color[rgb]{0,1,0}}%
      \expandafter\def\csname LT2\endcsname{\color[rgb]{0,0,1}}%
      \expandafter\def\csname LT3\endcsname{\color[rgb]{1,0,1}}%
      \expandafter\def\csname LT4\endcsname{\color[rgb]{0,1,1}}%
      \expandafter\def\csname LT5\endcsname{\color[rgb]{1,1,0}}%
      \expandafter\def\csname LT6\endcsname{\color[rgb]{0,0,0}}%
      \expandafter\def\csname LT7\endcsname{\color[rgb]{1,0.3,0}}%
      \expandafter\def\csname LT8\endcsname{\color[rgb]{0.5,0.5,0.5}}%
    \else
      % gray
      \def\colorrgb#1{\color{black}}%
      \def\colorgray#1{\color[gray]{#1}}%
      \expandafter\def\csname LTw\endcsname{\color{white}}%
      \expandafter\def\csname LTb\endcsname{\color{black}}%
      \expandafter\def\csname LTa\endcsname{\color{black}}%
      \expandafter\def\csname LT0\endcsname{\color{black}}%
      \expandafter\def\csname LT1\endcsname{\color{black}}%
      \expandafter\def\csname LT2\endcsname{\color{black}}%
      \expandafter\def\csname LT3\endcsname{\color{black}}%
      \expandafter\def\csname LT4\endcsname{\color{black}}%
      \expandafter\def\csname LT5\endcsname{\color{black}}%
      \expandafter\def\csname LT6\endcsname{\color{black}}%
      \expandafter\def\csname LT7\endcsname{\color{black}}%
      \expandafter\def\csname LT8\endcsname{\color{black}}%
    \fi
  \fi
    \setlength{\unitlength}{0.0500bp}%
    \ifx\gptboxheight\undefined%
      \newlength{\gptboxheight}%
      \newlength{\gptboxwidth}%
      \newsavebox{\gptboxtext}%
    \fi%
    \setlength{\fboxrule}{0.5pt}%
    \setlength{\fboxsep}{1pt}%
    \definecolor{tbcol}{rgb}{1,1,1}%
\begin{picture}(9060.00,3400.00)%
    \gplgaddtomacro\gplbacktext{%
      \csname LTb\endcsname%%
      \put(803,791){\makebox(0,0)[r]{\strut{}\footnotesize -2}}%
      \csname LTb\endcsname%%
      \put(803,1014){\makebox(0,0)[r]{\strut{}\footnotesize -1.5}}%
      \csname LTb\endcsname%%
      \put(803,1237){\makebox(0,0)[r]{\strut{}\footnotesize -1}}%
      \csname LTb\endcsname%%
      \put(803,1460){\makebox(0,0)[r]{\strut{}\footnotesize -0.5}}%
      \csname LTb\endcsname%%
      \put(803,1684){\makebox(0,0)[r]{\strut{}\footnotesize 0}}%
      \csname LTb\endcsname%%
      \put(803,1907){\makebox(0,0)[r]{\strut{}\footnotesize 0.5}}%
      \csname LTb\endcsname%%
      \put(803,2130){\makebox(0,0)[r]{\strut{}\footnotesize 1}}%
      \csname LTb\endcsname%%
      \put(803,2353){\makebox(0,0)[r]{\strut{}\footnotesize 1.5}}%
      \csname LTb\endcsname%%
      \put(803,2577){\makebox(0,0)[r]{\strut{}\footnotesize 2}}%
      \csname LTb\endcsname%%
      \put(904,436){\makebox(0,0){\strut{}\footnotesize 0}}%
      \csname LTb\endcsname%%
      \put(1523,436){\makebox(0,0){\strut{}\footnotesize 0.2}}%
      \csname LTb\endcsname%%
      \put(2142,436){\makebox(0,0){\strut{}\footnotesize 0.4}}%
      \csname LTb\endcsname%%
      \put(2761,436){\makebox(0,0){\strut{}\footnotesize 0.6}}%
      \csname LTb\endcsname%%
      \put(3380,436){\makebox(0,0){\strut{}\footnotesize 0.8}}%
      \csname LTb\endcsname%%
      \put(4000,436){\makebox(0,0){\strut{}\footnotesize 1}}%
    }%
    \gplgaddtomacro\gplfronttext{%
      \csname LTb\endcsname%%
      \put(2452,76){\makebox(0,0){\strut{}\small $x$}}%
    }%
    \gplgaddtomacro\gplbacktext{%
      \csname LTb\endcsname%%
      \put(4939,930){\makebox(0,0)[r]{\strut{}\footnotesize $10^{-6}$}}%
      \csname LTb\endcsname%%
      \put(4939,1336){\makebox(0,0)[r]{\strut{}\footnotesize $10^{-5}$}}%
      \csname LTb\endcsname%%
      \put(4939,1742){\makebox(0,0)[r]{\strut{}\footnotesize $10^{-4}$}}%
      \csname LTb\endcsname%%
      \put(4939,2148){\makebox(0,0)[r]{\strut{}\footnotesize $10^{-3}$}}%
      \csname LTb\endcsname%%
      \put(4939,2555){\makebox(0,0)[r]{\strut{}\footnotesize $10^{-2}$}}%
      \csname LTb\endcsname%%
      \put(5039,436){\makebox(0,0){\strut{}\footnotesize 0}}%
      \csname LTb\endcsname%%
      \put(5659,436){\makebox(0,0){\strut{}\footnotesize 50}}%
      \csname LTb\endcsname%%
      \put(6279,436){\makebox(0,0){\strut{}\footnotesize 100}}%
      \csname LTb\endcsname%%
      \put(6899,436){\makebox(0,0){\strut{}\footnotesize 150}}%
      \csname LTb\endcsname%%
      \put(7519,436){\makebox(0,0){\strut{}\footnotesize 200}}%
    }%
    \gplgaddtomacro\gplfronttext{%
      \csname LTb\endcsname%%
      \put(7587,2488){\makebox(0,0)[r]{\strut{}\footnotesize $S(\omg)$}}%
      \csname LTb\endcsname%%
      \put(7587,2248){\makebox(0,0)[r]{\strut{}\footnotesize our estimator}}%
      \csname LTb\endcsname%%
      \put(7587,2008){\makebox(0,0)[r]{\strut{}\footnotesize interp. estimator}}%
      \csname LTb\endcsname%%
      \put(6587,76){\makebox(0,0){\strut{}\small $\omg$}}%
    }%
    \gplbacktext
    \put(0,0){\includegraphics[width={453.00bp},height={170.00bp}]{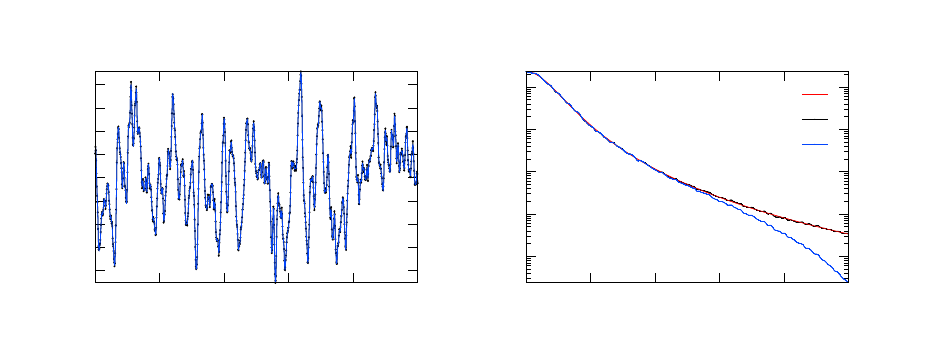}}%
    \gplfronttext
  \end{picture}%
\endgroup
  \caption{An example spectral density estimator on grid-interpolated data with
  $1{,}000$ replicates.  Left panel: raw data (black dots) and the interpolating
  Gaussian kernel-based gridded estimator (blue lines). Right panel: spectral
  density estimators for each case, with the true spectral density shown in
  red.}
  \label{fig:interp}
\end{figure}

Figure \ref{fig:interp} gives a demonstration of this phenomenon. Here, data are
sampled at a lightly jittered grid of $n=1{,}000$ points on $[0,1]$, with the
jitter uniform in the range $[-5\cdot10^{-5}, 5\cdot10^{-5}]$. The true process
simulated at those points is a Mat\'ern process with $\nu = 3/2$, which is once
mean-square differentiable but not exceptionally smooth. The process was
interpolated to a regular grid with the Gaussian kernel $k(x,y) = e^{-|x-y|^2 /
0.001}$ and $10$ nearest neighbors, where the smaller range parameter was
selected in an attempt to make $k(x,y)$ agree as well as possible with the true
kernel around the origin. The left panel of Figure \ref{fig:interp} shows one
replicate of the true data with the interpolated values at regular grid points
overlaid, demonstrating that the disagreement between the two plotted lines is
imperceptible to the eye and that one might very plausibly conclude that the
regridding did not introduce any of the standard artifacts of miscalibrated
interpolation. The right panel, however, shows the spectral density estimator of
the re-gridded measurements next to the direct estimator we propose here
computed directly from the irregularly sampled data. Here the disagreement is
extremely clear: the Gaussian kernel interpolation has hugely impacted the
estimated power at moderate to high frequencies, and the inflection point where
the gridded estimator begins to decay like a Gaussian spectral density, as
opposed to a Mat\'ern spectral density, is clear.

\subsection{Two-dimensional processes}

Moving to two dimensions, we now consider an analogous problem of estimating the
spectral density $S$ for a process $Y(\bx)$ defined on $[0,1]^2$ measured at $n
= 10^6$ random locations $\bx_j \iid \text{Unif}([0,1]^2)$. As mentioned in
Section \ref{sec:method}, no aspect of our method is specific to one dimension,
and our fundamental quadrature-based observation that one can design weights
$\bal$ such that the Fourier sum $H_{\alpha}(\bomg)$ approximates the continuous
Fourier transform of a window function remains valid. Letting $\bm{\Omega} =
(\Omega_1, ..., \Omega_d)$ denote the multivariate analog of $\Omega$, we note
that the maximum controllable frequency in each dimension can grow at most like
$\Omega_j = \bO(n^{1/d})$ when the total number of measurements is $n$, just as
in the gridded data case for a process on a hypercube. As such, even in two
dimensions the frequencies that can be resolved with large amounts of data can
appear relatively small.

\begin{figure}[] 
  \centering 
% GNUPLOT: LaTeX picture with Postscript
\begingroup
  \makeatletter
  \providecommand\color[2][]{%
    \GenericError{(gnuplot) \space\space\space\@spaces}{%
      Package color not loaded in conjunction with
      terminal option `colourtext'%
    }{See the gnuplot documentation for explanation.%
    }{Either use 'blacktext' in gnuplot or load the package
      color.sty in LaTeX.}%
    \renewcommand\color[2][]{}%
  }%
  \providecommand\includegraphics[2][]{%
    \GenericError{(gnuplot) \space\space\space\@spaces}{%
      Package graphicx or graphics not loaded%
    }{See the gnuplot documentation for explanation.%
    }{The gnuplot epslatex terminal needs graphicx.sty or graphics.sty.}%
    \renewcommand\includegraphics[2][]{}%
  }%
  \providecommand\rotatebox[2]{#2}%
  \@ifundefined{ifGPcolor}{%
    \newif\ifGPcolor
    \GPcolortrue
  }{}%
  \@ifundefined{ifGPblacktext}{%
    \newif\ifGPblacktext
    \GPblacktexttrue
  }{}%
  % define a \g@addto@macro without @ in the name:
  \let\gplgaddtomacro\g@addto@macro
  % define empty templates for all commands taking text:
  \gdef\gplbacktext{}%
  \gdef\gplfronttext{}%
  \makeatother
  \ifGPblacktext
    % no textcolor at all
    \def\colorrgb#1{}%
    \def\colorgray#1{}%
  \else
    % gray or color?
    \ifGPcolor
      \def\colorrgb#1{\color[rgb]{#1}}%
      \def\colorgray#1{\color[gray]{#1}}%
      \expandafter\def\csname LTw\endcsname{\color{white}}%
      \expandafter\def\csname LTb\endcsname{\color{black}}%
      \expandafter\def\csname LTa\endcsname{\color{black}}%
      \expandafter\def\csname LT0\endcsname{\color[rgb]{1,0,0}}%
      \expandafter\def\csname LT1\endcsname{\color[rgb]{0,1,0}}%
      \expandafter\def\csname LT2\endcsname{\color[rgb]{0,0,1}}%
      \expandafter\def\csname LT3\endcsname{\color[rgb]{1,0,1}}%
      \expandafter\def\csname LT4\endcsname{\color[rgb]{0,1,1}}%
      \expandafter\def\csname LT5\endcsname{\color[rgb]{1,1,0}}%
      \expandafter\def\csname LT6\endcsname{\color[rgb]{0,0,0}}%
      \expandafter\def\csname LT7\endcsname{\color[rgb]{1,0.3,0}}%
      \expandafter\def\csname LT8\endcsname{\color[rgb]{0.5,0.5,0.5}}%
    \else
      % gray
      \def\colorrgb#1{\color{black}}%
      \def\colorgray#1{\color[gray]{#1}}%
      \expandafter\def\csname LTw\endcsname{\color{white}}%
      \expandafter\def\csname LTb\endcsname{\color{black}}%
      \expandafter\def\csname LTa\endcsname{\color{black}}%
      \expandafter\def\csname LT0\endcsname{\color{black}}%
      \expandafter\def\csname LT1\endcsname{\color{black}}%
      \expandafter\def\csname LT2\endcsname{\color{black}}%
      \expandafter\def\csname LT3\endcsname{\color{black}}%
      \expandafter\def\csname LT4\endcsname{\color{black}}%
      \expandafter\def\csname LT5\endcsname{\color{black}}%
      \expandafter\def\csname LT6\endcsname{\color{black}}%
      \expandafter\def\csname LT7\endcsname{\color{black}}%
      \expandafter\def\csname LT8\endcsname{\color{black}}%
    \fi
  \fi
    \setlength{\unitlength}{0.0500bp}%
    \ifx\gptboxheight\undefined%
      \newlength{\gptboxheight}%
      \newlength{\gptboxwidth}%
      \newsavebox{\gptboxtext}%
    \fi%
    \setlength{\fboxrule}{0.5pt}%
    \setlength{\fboxsep}{1pt}%
    \definecolor{tbcol}{rgb}{1,1,1}%
\begin{picture}(9060.00,3960.00)%
    \gplgaddtomacro\gplbacktext{%
      \csname LTb\endcsname%%
      \put(351,710){\makebox(0,0)[r]{\strut{}\footnotesize 0}}%
      \csname LTb\endcsname%%
      \put(351,1385){\makebox(0,0)[r]{\strut{}\footnotesize 0.33}}%
      \csname LTb\endcsname%%
      \put(351,2061){\makebox(0,0)[r]{\strut{}\footnotesize 0.66}}%
      \csname LTb\endcsname%%
      \put(351,2736){\makebox(0,0)[r]{\strut{}\footnotesize 0.99}}%
      \csname LTb\endcsname%%
      \put(453,469){\makebox(0,0){\strut{}\footnotesize 0}}%
      \csname LTb\endcsname%%
      \put(1148,469){\makebox(0,0){\strut{}\footnotesize 0.33}}%
      \csname LTb\endcsname%%
      \put(1843,469){\makebox(0,0){\strut{}\footnotesize 0.66}}%
      \csname LTb\endcsname%%
      \put(2539,469){\makebox(0,0){\strut{}\footnotesize 0.99}}%
    }%
    \gplgaddtomacro\gplfronttext{%
      \csname LTb\endcsname%%
      \put(1506,157){\makebox(0,0){\strut{}}}%
      \csname LTb\endcsname%%
      \put(1506,2961){\makebox(0,0){\strut{}\footnotesize regridded sample path}}%
    }%
    \gplgaddtomacro\gplbacktext{%
      \csname LTb\endcsname%%
      \put(3364,965){\makebox(0,0)[r]{\strut{}\footnotesize -150}}%
      \csname LTb\endcsname%%
      \put(3364,1733){\makebox(0,0)[r]{\strut{}\footnotesize 0}}%
      \csname LTb\endcsname%%
      \put(3364,2501){\makebox(0,0)[r]{\strut{}\footnotesize 150}}%
      \csname LTb\endcsname%%
      \put(3729,469){\makebox(0,0){\strut{}\footnotesize -150}}%
      \csname LTb\endcsname%%
      \put(4520,469){\makebox(0,0){\strut{}\footnotesize 0}}%
      \csname LTb\endcsname%%
      \put(5310,469){\makebox(0,0){\strut{}\footnotesize 150}}%
    }%
    \gplgaddtomacro\gplfronttext{%
      \csname LTb\endcsname%%
      \put(2967,1733){\rotatebox{-270.00}{\makebox(0,0){\strut{}\footnotesize $\omega_2$}}}%
      \csname LTb\endcsname%%
      \put(4519,110){\makebox(0,0){\strut{}\footnotesize $\omega_1$}}%
      \csname LTb\endcsname%%
      \put(4519,2961){\makebox(0,0){\strut{}\footnotesize true $S(\bomg)$}}%
    }%
    \gplgaddtomacro\gplbacktext{%
      \csname LTb\endcsname%%
      \put(6742,469){\makebox(0,0){\strut{}\footnotesize -150}}%
      \csname LTb\endcsname%%
      \put(7533,469){\makebox(0,0){\strut{}\footnotesize 0}}%
      \csname LTb\endcsname%%
      \put(8324,469){\makebox(0,0){\strut{}\footnotesize 150}}%
    }%
    \gplgaddtomacro\gplfronttext{%
      \csname LTb\endcsname%%
      \put(7533,110){\makebox(0,0){\strut{}\footnotesize $\omega_1$}}%
      \csname LTb\endcsname%%
      \put(8846,855){\makebox(0,0)[l]{\strut{}\footnotesize 1e-14}}%
      \csname LTb\endcsname%%
      \put(8846,1148){\makebox(0,0)[l]{\strut{}\footnotesize 1e-12}}%
      \csname LTb\endcsname%%
      \put(8846,1440){\makebox(0,0)[l]{\strut{}\footnotesize 1e-10}}%
      \csname LTb\endcsname%%
      \put(8846,1733){\makebox(0,0)[l]{\strut{}\footnotesize 1e-08}}%
      \csname LTb\endcsname%%
      \put(8846,2026){\makebox(0,0)[l]{\strut{}\footnotesize 1e-06}}%
      \csname LTb\endcsname%%
      \put(8846,2318){\makebox(0,0)[l]{\strut{}\footnotesize 0.0001}}%
      \csname LTb\endcsname%%
      \put(8846,2611){\makebox(0,0)[l]{\strut{}\footnotesize 0.01}}%
      \csname LTb\endcsname%%
      \put(7533,2961){\makebox(0,0){\strut{}\footnotesize $\E \hat{S}(\bomg)$}}%
    }%
    \gplbacktext
    \put(0,0){\includegraphics[width={453.00bp},height={198.00bp}]{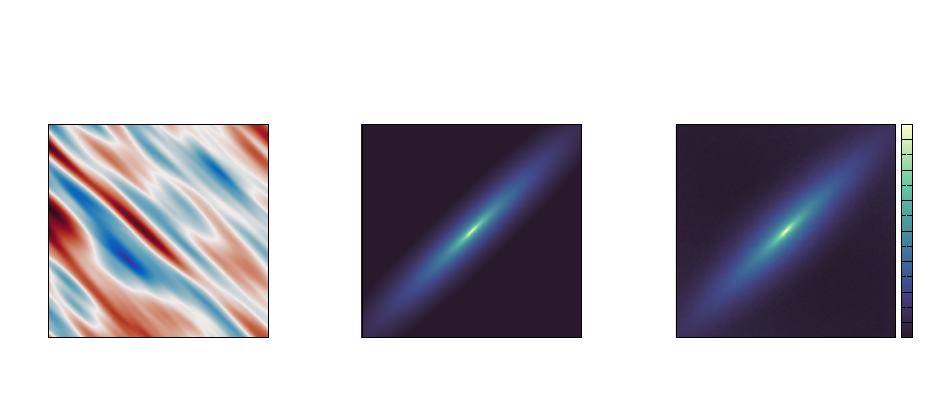}}%
    \gplfronttext
  \end{picture}%
\endgroup
  \caption{A two-dimensional nonparametric spectral density estimator for $n=1$M
  samples from an anisotropic Mat\'ern process on $[0,1]^2$. Left panel: a
  sample path from the process on a regular grid. Center panel: the true
  spectral density $S(\bomg)$. Right panel: the expected value $\E
  \hat{S}(\bomg)$ approximated with $500$ Monte Carlo samples.}
  \label{fig:npest2d} 
\end{figure}

However, we can combat this curse of dimensionality, as our estimator can be
computed rapidly for very large data sets using the accelerated methods
introduced in Section~\ref{sec:method}. Figure \ref{fig:npest2d} gives an
example of the estimator applied to $n=10^6$ samples of an \emph{anisotropic}
(meaning that its spectral density is not radially symmetric) Mat\'ern process
on $[0,1]^2$. Using the sparse preconditioner described in Section
\ref{sec:numerical}, this estimator was computed in approximately $4.5$ minutes
on an AMD EPYC $9554$P processor using $40$ cores. The window function $g$ used
for the estimator is a two-dimensional prolate function where the domain of
support $\mathcal{R}$ for $G$ is a disk of radius $5$, and as the right panel of
the figure demonstrates this estimator armed with such a window $g$ is capable
of capturing $14$ orders of magnitude of variation. For smaller data sizes,
since $\Omega_j$ grows slowly and the radius of the spectral support region
$\mathcal{R}$ is hard to reduce, artifacts can be significant. But as this
numerical experiment demonstrates, spectral density estimation in multiple
dimensions is not necessarily cursed to always have significant finite-sample or
window-based artifacts --- it simply requires much more data than one is
accustomed to needing in one-dimensional settings.

\section{Discussion} \label{sec:discussion}

In this work, we present a framework for nonparametric spectral density
estimation that applies to fully irregularly sampled measurements in one or more
dimensions. The key insight is to implicitly obtain quadrature weights such that
the variance of the weighted nonuniform Fourier sums can be interpreted as
integral transforms of the SDF in a way that mimics the standard univariate
gridded setting. This estimator demonstrates significant practical improvement
over periodogram or Lomb-Scargle least squares-based approaches. However, there
are also limitations to its accuracy due to aliasing-based sources of bias in
the tails of the spectral density, and there is a direct tradeoff between how
far in the tails of an SDF one hopes to estimate and the chance of aliasing
biases dominating the variance of the true signal one is attempting to observe.
Various theoretical and computational aspects of selecting the window $g$ and
scalably computing the corresponding weights $\bal$ remain exciting avenues for
future research.

\ifbool{anon}{}{
  \renewcommand{\abstractname}{Acknowledgments}
  \begin{abstract}
  \noindent The authors would like to thank Rishabh Dudeja for suggesting the
  Jensen inequality mechanism used in the proof of Theorem $2$, and
  Michael O'Neil for helpful conversations throughout the development process.
  \end{abstract}
}

\bibliography{references}{}

\newpage

\appendix

\section{Proofs} \label{app:proofs}
\subsection*{Proof of Theorem~\ref{thm:recovery}}
\begin{proof}
    Let $\{\omega_k\}_{k=1}^n$ be Chebyshev nodes on $[-\Omega,\Omega]$. Then
    the linear system
    \begin{equation}
    H_{\bm{\alpha}}(\omega_k) 
    = \sum_{j=1}^n \alpha_j e^{-2\pi i\omega_k x_j}
    = G(\omega_k), \qquad 
    k = 1,\dots,n
    \end{equation}
    has a unique solution $\bm{\alpha} \in \C^n$ for which the order-$n$
    Chebyshev interpolants of $G$ and the resulting $H_{\bm{\alpha}}$ are
    everywhere equal. $H_{\bal}$ and $G$ are band-limited and can thus be
    analytically extended to entire functions, and therefore their pointwise
    error converges to zero exponentially as
    \begin{equation} \label{eq:chebyshev-convergence}
    \norm[{L^\infty([-\Omega,\Omega])}]{H_{\bm{\alpha}} - G}
    \leq \frac{4M(\rho)\rho^{-n}}{\rho - 1}
    \end{equation}
    for all $\rho > 1$, where $\abs{H_{\bm{\alpha}}(z) - G(z)} \leq M(\rho)$ for
    all $z \in E_\rho$ and $E_\rho := \{(z + z^{-1})/2 : |z| = \rho \}$ is the
    $\rho$-Bernstein ellipse~\cite[Theorem 8.2]{trefethen2019approximation}.
    Defining $F(z) := H_{\bm{\alpha}}(z) - G(z)$, we can compute the explicit
    bound 
    \begin{align*}
    \abs{F(z)} 
    &= \abs{F(p + iq)} \\
    &= \abs{\int_{-\frac{\Omega}{2}(b-a)}^{\frac{\Omega}{2}(b-a)} \left(\sum_{j=1}^n \alpha_j \delta(x-x_j) - g(x)\right) e^{-2\pi i (p + iq)x} \dif{x}} \\
    &\leq \int_{-\frac{\Omega}{2}(b-a)}^{\frac{\Omega}{2}(b-a)} \left(\sum_{j=1}^n \abs{\alpha_j} \delta(x-x_j) + \abs{g(x)}\right) e^{2\pi qx} \dif{x} \\
    &\leq e^{\pi \Omega (b-a) \abs{z}} \left( \norm[{L^1([a,b])}]{g} + \norm[1]{\bm{\alpha}} \right)
    \end{align*}
    The point in $E_\rho$ with largest magnitude is $z = \frac{\rho +
    \rho^{-1}}{2}$. Thus we have
    \begin{equation*}
    \abs{F(z)}
    \leq \exp\left\{\frac{\pi}{2} \Omega (b-a) (\rho + \rho^{-1})\right\} \left( \norm[{L^1([a,b])}]{g} + \norm[1]{\bm{\alpha}} \right) =: M(\rho)
    \end{equation*}
    for all $z \in E_\rho$. The result then follows
    from~\eqref{eq:chebyshev-convergence}.
\end{proof}

\subsection*{Proof of Theorem~\ref{thm:aliasing}}

\begin{lemma} \label{lemma:dist} Let $\set{x_j}_{j=1}^n \iid p$ with $p$
  satisfying (A$2$) and $\bm{\alpha} = [\alpha_j]_{j=1}^n$ be weights satisfying
  (A$5$). Then by a standard abuse of notation
  \begin{equation*} %\label{eq:}
    H_{\alpha}(\omg) = \sum_{j=1}^n e^{2 \pi i \omg x_j} \alpha_j
    \cvd
    \mathcal{C}\Nd\set{
      \mat{G_0 \varphi(\omg) \\ 0},
      \;
      \bS
    },
  \end{equation*}
  where $\mathcal{C}\Nd$ denotes the complex-normal distribution, $\bm{\alpha} =
  \bm{a} + i \bm{b}$, $\inp{\cdot, \cdot}$ denotes the Euclidean inner product,
  and $\bS$ has entries given by
  \begin{align*} %\label{eq:}
    \bS_{11} &= 
    \norm[2]{\bm{a}}^2 \left( \frac{1 + \varphi(2 \omg)}{2} -
    \varphi(\omg)^2 \right)
    +
    \norm[2]{\bm{b}}^2 \frac{1 - \varphi(2 \omg)}{2}
    \\
    \bS_{12} = \bS_{21} &= 
    \inp{\bm{a}, \bm{b}} (\varphi(2 \omg) - \varphi(\omg)^2).
    \\
    \bS_{22} &= 
    \norm[2]{\bm{a}}^2 \frac{1 - \varphi(2 \omg)}{2}
    +
    \norm[2]{\bm{b}}^2 \left(\frac{1 + \varphi(2 \omg)}{2} -
    \varphi(\omg)^2 \right).
  \end{align*}
  Additionally, as $\omg \to \infty$ as well, we have the simplified limiting
  MGF for $\abs{H_{\alpha}(\omg)}^2$ given by
  \begin{equation} \label{eq:H_mgf}
    \E e^{t |H_{\alpha}(\omg)|^2}
    \to_{n, \omg}
    \frac{1}{1 - t \norm[2]{\bal}^2}
  \end{equation}
  for $t$ such that the right-hand side is finite.
\end{lemma}

\begin{proof}
  First, we break $H_{\alpha}(\omg)$ into its real and imaginary parts
  \begin{align*} %\label{eq:}
    H_{\alpha}^r(\omg) &= 
    \sum_{j=1}^n \cos(2 \pi \omg x_j) a_j - \sin(2 \pi \omg x_j) b_j
    \\
    H_{\alpha}^i(\omg) &= 
    \sum_{j=1}^n \sin(2 \pi \omg x_j) a_j + \cos(2 \pi \omg x_j) b_j.
  \end{align*}
  For the first moments, we simply note that $\varphi(\omg)$ is real-valued by
  (A$2$) and we have that $\E H_{\alpha}^r(\omg) = G_0 \varphi(\omg)$ and $\E
  H_{\alpha}^i(\omg) = 0$. For the second moments, the fundamental computations
  required are $\V \cos(2 \pi \omg x_j)$, $\V \sin(2 \pi \omg x_j)$, and
  $\text{Cov}(\cos(2 \pi \omg x_j), \sin(2 \pi \omg x_j))$. In all cases, these
  may be expressed using standard product-to-sum trigonometric formulae. As an
  example, we note that 
  \begin{align*} %\label{eq:}
    \V \cos(2 \pi \omg x_j) &= \E \cos(2 \pi \omg x_j)^2 - (\E \cos(2 \pi \omg
    x_j))^2
    \\
    &= \E\set{
      \frac{1}{2}[1 + \cos(2 \pi (2 \omg) x_j)]
    }
    -
    \varphi(\omg)^2
    \\
    &= \frac{1}{2}(1 + \varphi(2 \omg)) - \varphi(\omg)^2.
  \end{align*}
  By similar arguments, we see that $\V \sin(2 \pi \omg x_j) = \frac{1}{2}(1 -
  \phi(2 \omg))$ and that the covariance of the cross term is $\text{Cov}(\cos(2
  \pi \omg x_j), \sin(2 \pi \omg x_j)) = 0$. With this established, we first
  compute the marginal variance of $H_{\alpha}^r(\omg)$ as
  \begin{align*} %\label{eq:}
    \V H_{\alpha}^r(\omg)
    &= \V\set{\sum_{j=1}^n \cos(2 \pi \omg x_j) a_j}
    +
    \V\set{\sum_{j=1}^n \sin(2 \pi \omg x_j) b_j}
    \\
    &= 
    \norm[2]{\bm{a}}^2 \left( \frac{1 + \varphi(2 \omg)}{2} -
    \varphi(\omg)^2 \right)
    +
    \norm[2]{\bm{b}}^2 \frac{1 - \varphi(2 \omg)}{2}
  \end{align*}
  where the cross term is zero by above and the independence of the $\set{x_j}$
  is used to bring the variance inside the sums. Similarly, we have that the
  marginal variance for $H_{\alpha}^i(\omg)$ is given as
  \begin{equation*} %\label{eq:}
    \V H_{\alpha}^i(\omg)
    =
    \norm[2]{\bm{a}}^2 \frac{1 - \varphi(2 \omg)}{2}
    +
    \norm[2]{\bm{b}}^2 \left(\frac{1 + \varphi(2 \omg)}{2} -
    \varphi(\omg)^2 \right).
  \end{equation*}
  Finally, by similar or straightforward definitional computations, we see that
  \begin{align*} %\label{eq:}
    \text{Cov}(H_{\alpha}^r(\omg), H_{\alpha}^i(\omg))
    = 
    \inp{\bm{a}, \bm{b}} (\varphi(2 \omg) - \varphi(\omg)^2).
  \end{align*}
  With the first and second moments computed, the first claim follows from the
  Lyapunov central limit theorem.

  For the second claim, we note that for each $\omg$ and $n$,
  $|H_{\alpha}(\omg)| \leq \norm[1]{\bm{\alpha}} = \bO(1)$ by (A$5$), so that
  the doubly-indexed collection of random variables is uniformly bounded and
  thus asymptotically uniformly integrable, and the same is true for the family
  of random variables $e^{t |H_{\alpha}(\omg)|^2}$. By standard results on
  uniform integrability (for example \cite[Theorem 2.20]{van2000}), we have that
  $e^{t |H_{\alpha}(\omg)|^2}$ converges in expectation as well as in law. For a
  random variable $\bm{s} \sim \Nd\set{\bm{0}, \bm{C}}$, the moment generating
  function of $\bm{s}^T \bm{s}$ is given by
  \begin{equation} \label{eq:mvn_qform_mgf}
    \E e^{t \bm{s}^T \bm{s}} = \abs{\I - 2t \bm{C}}^{-1/2},
  \end{equation}
  and applying in this case to the asymptotic variance $\bm{C} = \bS$ with $\omg
  \to 0$ reduces $\bm{C}$ to $\norm[2]{\bal}^2 \I$. Plugging this reduction in
  (\ref{eq:mvn_qform_mgf}) gives the desired result.
\end{proof} 

\begin{proof}[Proof of Theorem~\ref{thm:aliasing}]
  Letting $q(x)$ be a valid density on $E_{\Omega, \xi}$, Jensen's inequality
  gives that
  \begin{align*} %\label{eq:}
    \E e^{t \eps(\xi)} 
    \leq 
    \int_{E_{\Omega, \xi}}
    \E e^{t \frac{S(\omg)}{q(\omg)} \abs{H_{\alpha}(\omg - \xi)}^2 }
    q(\omg) \dif \omg.
  \end{align*}
  Since $S(\omg)$ is a spectral density, we note that $q(\omg) = C S(\omg)$ is a
  valid probability density for $C = \left(\int_{E_{\Omega, \xi}} S(\omg) \dif
  \omg \right)^{-1}$. Substituting this choice of $q$ gives
  \begin{align*} %\label{eq:}
    \E e^{t \eps(\xi)} 
    &\leq 
    C
    \int_{E_{\Omega, \xi}} 
    M_{|H_{\alpha}(\omg - \xi)|^2}(t/C) S(\omg) \dif \omg \\
    &\approx
    \frac{C}{1 - \tfrac{t}{C} \norm[2]{\bm{\alpha}}^2}
    \int_{E_{\Omega, \xi}} 
    S(\omg) \dif \omg \\
    &=
    \frac{1}{1 - \tfrac{t}{C} \norm[2]{\bm{\alpha}}^2}.
  \end{align*}
  Where we have used Lemma~\ref{lemma:dist} regarding the moment generating
  function $M_{|H_{\alpha}(\omg - \xi)|^2}$. With this bound on the MGF
  established, by a Chernoff bound with $t = \frac{\delta C}{\norm[2]{\bal}^2}$,
  $\delta \in [0,1)$, we see that
  \begin{align*} %\label{eq:}
   P\set{\eps(\xi) \geq \beta S(\xi)}
   \leq
   \frac{
      \exp\set{-t \beta S(\xi)}
    }{
      1 - \tfrac{t}{C} \norm[2]{\bal}^2
    }
    =
    (1 - \delta)^{-1} \exp\set{
    -
      \frac{
        \delta \beta S(\xi)
      }{
        \norm[2]{\bal}^2 \int_{E_{\Omega, \xi}} S(\omg) \dif \omg
      }
    }.
  \end{align*}
  Picking the specific case of $\delta=1/2$ completes the proof.
\end{proof}

\subsection*{Proof of Corollary~\ref{cor:rate}}
\begin{proof}
  By assumption (A$5$), $\sup_n \norm[1]{\bal(n)} < \infty$, and so $\sup_n
  \norm[2]{\bal(n)} = C_{\alpha} < \infty$ gives a uniform bound on the weight
  norms as $n$ grows. Using the survival function form of the expectation and a
  slight variation on the stated form of Theorem \ref{thm:aliasing}, one has
  that
  \begin{align*} %\label{eq:}
    \E \eps(\xi_n) 
    =
    \int_{0}^{\infty} P \set{\eps(\xi_n) > t} \dif t
    \leq
    \int_{0}^{\infty} \exp \set{\frac{-t}{2 C_{\alpha}
    \int_{E_{\Omega_n, \xi_n}} S(\omg) \dif \omg}} \dif t.
  \end{align*}
  By the assumption that $n$ is sufficiently large that $S(\omg) \approx C_0
  |\omg|^{-2 \nu - 1}$ on the super-Nyquist region $E_{\Omega_n, \xi_n}$ and
  that $\xi_n \geq 0$, one has that
  \begin{equation*} %\label{eq:}
    \int_{E_{\Omega_n, \xi_n}} S(\omg) \dif \omg
    \leq
    2 C_0 \int_{-\infty}^{-\Omega_n + \xi_n} |\omg|^{-2 \nu - 1} \dif \omg
    =
    2 C_0 \frac{ (\Omega_n - \xi_n)^{-2 \nu} }{ 2 \nu }.
  \end{equation*}
  Subtituting this into the previous bound and applying the assumed forms for
  $\Omega_n$ and $\xi_n$ gives that 
  \begin{equation*} %\label{eq:}
    \E \eps(\xi_n) 
    \leq
    \int_{0}^{\infty} \exp \set{\frac{-t \nu}{
    2 C_{\alpha} C_0 (\Omega_n - \xi_n)^{-2 \nu} 
    }} \dif t
    =
    \frac{2 C_{\alpha} C_0 (\Omega_n - \xi_n)^{-2 \nu}}{\nu},
  \end{equation*}
  as postulated.
\end{proof}

\section{Appendix 3: Additional numerical experiments}
\subsection*{Aliasing bias and the norm of the weights}

As a second demonstration of the role of $\norm[2]{\bal}$ in controlling
aliasing bias, we revisit the setting of Figure~\ref{fig:gap_prolate}, this time
picking sampling locations $\{x_j\}_{j=1}^n \iid \text{Unif}([-1, 1] \setminus
[-\tau, 0])$, simulating a Mat\'ern process with spectral density
$
  S(\omg) = C(\nu, \rho) \sgm^2 (2 \nu/\rho^2 + 4 \pi^2 \omg^2)^{-\nu - 1/2}
$
\cite{stein1999} and parameters $(\sgm, \rho, \nu) = (1, 0.1, 0.75)$, and
attempting to estimate $S$. This process exhibits strong dependence and slow
spectral decay, and represents a setting where aliasing bias is a particularly
serious concern. Figure \ref{fig:norm_est} shows the result of choosing $g$ to
be either a standard Kaiser window supported on $[-1,1]$ or a prolate function
supported on the union of the two disjoint sampling intervals for several values
of $\tau$. As the figure demonstrates, for small enough $\tau$ the two
estimators behave reasonably similarly, but as the gap $\tau$ increases and the
norm $\norm[2]{\bal}$ increases the estimates $\hat{S}(\xi)$ eventually are
uniformly ruined by the size of the aliasing errors. We emphasize, however, that
even with the moderate gap the Kaiser-based $G$ \emph{is} recovered reasonably
well by the weights $\bal$ as demonstrated in Figure~\ref{fig:gap_prolate}. So
the failure here is not that there is no $\bal$ such that $\Frt \bal \approx
\bm{b}$, but rather that the $\bal$ which provides the best approximate solution
to that system has large norm.

\begin{figure}[!ht]
  \centering
% GNUPLOT: LaTeX picture with Postscript
\begingroup
  \makeatletter
  \providecommand\color[2][]{%
    \GenericError{(gnuplot) \space\space\space\@spaces}{%
      Package color not loaded in conjunction with
      terminal option `colourtext'%
    }{See the gnuplot documentation for explanation.%
    }{Either use 'blacktext' in gnuplot or load the package
      color.sty in LaTeX.}%
    \renewcommand\color[2][]{}%
  }%
  \providecommand\includegraphics[2][]{%
    \GenericError{(gnuplot) \space\space\space\@spaces}{%
      Package graphicx or graphics not loaded%
    }{See the gnuplot documentation for explanation.%
    }{The gnuplot epslatex terminal needs graphicx.sty or graphics.sty.}%
    \renewcommand\includegraphics[2][]{}%
  }%
  \providecommand\rotatebox[2]{#2}%
  \@ifundefined{ifGPcolor}{%
    \newif\ifGPcolor
    \GPcolortrue
  }{}%
  \@ifundefined{ifGPblacktext}{%
    \newif\ifGPblacktext
    \GPblacktexttrue
  }{}%
  % define a \g@addto@macro without @ in the name:
  \let\gplgaddtomacro\g@addto@macro
  % define empty templates for all commands taking text:
  \gdef\gplbacktext{}%
  \gdef\gplfronttext{}%
  \makeatother
  \ifGPblacktext
    % no textcolor at all
    \def\colorrgb#1{}%
    \def\colorgray#1{}%
  \else
    % gray or color?
    \ifGPcolor
      \def\colorrgb#1{\color[rgb]{#1}}%
      \def\colorgray#1{\color[gray]{#1}}%
      \expandafter\def\csname LTw\endcsname{\color{white}}%
      \expandafter\def\csname LTb\endcsname{\color{black}}%
      \expandafter\def\csname LTa\endcsname{\color{black}}%
      \expandafter\def\csname LT0\endcsname{\color[rgb]{1,0,0}}%
      \expandafter\def\csname LT1\endcsname{\color[rgb]{0,1,0}}%
      \expandafter\def\csname LT2\endcsname{\color[rgb]{0,0,1}}%
      \expandafter\def\csname LT3\endcsname{\color[rgb]{1,0,1}}%
      \expandafter\def\csname LT4\endcsname{\color[rgb]{0,1,1}}%
      \expandafter\def\csname LT5\endcsname{\color[rgb]{1,1,0}}%
      \expandafter\def\csname LT6\endcsname{\color[rgb]{0,0,0}}%
      \expandafter\def\csname LT7\endcsname{\color[rgb]{1,0.3,0}}%
      \expandafter\def\csname LT8\endcsname{\color[rgb]{0.5,0.5,0.5}}%
    \else
      % gray
      \def\colorrgb#1{\color{black}}%
      \def\colorgray#1{\color[gray]{#1}}%
      \expandafter\def\csname LTw\endcsname{\color{white}}%
      \expandafter\def\csname LTb\endcsname{\color{black}}%
      \expandafter\def\csname LTa\endcsname{\color{black}}%
      \expandafter\def\csname LT0\endcsname{\color{black}}%
      \expandafter\def\csname LT1\endcsname{\color{black}}%
      \expandafter\def\csname LT2\endcsname{\color{black}}%
      \expandafter\def\csname LT3\endcsname{\color{black}}%
      \expandafter\def\csname LT4\endcsname{\color{black}}%
      \expandafter\def\csname LT5\endcsname{\color{black}}%
      \expandafter\def\csname LT6\endcsname{\color{black}}%
      \expandafter\def\csname LT7\endcsname{\color{black}}%
      \expandafter\def\csname LT8\endcsname{\color{black}}%
    \fi
  \fi
    \setlength{\unitlength}{0.0500bp}%
    \ifx\gptboxheight\undefined%
      \newlength{\gptboxheight}%
      \newlength{\gptboxwidth}%
      \newsavebox{\gptboxtext}%
    \fi%
    \setlength{\fboxrule}{0.5pt}%
    \setlength{\fboxsep}{1pt}%
    \definecolor{tbcol}{rgb}{1,1,1}%
\begin{picture}(9060.00,3400.00)%
    \gplgaddtomacro\gplbacktext{%
      \csname LTb\endcsname%%
      \put(1255,998){\makebox(0,0)[r]{\strut{}\footnotesize $10^{-1}$}}%
      \csname LTb\endcsname%%
      \put(1255,1459){\makebox(0,0)[r]{\strut{}\footnotesize $10^{0}$}}%
      \csname LTb\endcsname%%
      \put(1255,1920){\makebox(0,0)[r]{\strut{}\footnotesize $10^{1}$}}%
      \csname LTb\endcsname%%
      \put(1255,2381){\makebox(0,0)[r]{\strut{}\footnotesize $10^{2}$}}%
      \csname LTb\endcsname%%
      \put(1743,436){\makebox(0,0){\strut{}\footnotesize 0}}%
      \csname LTb\endcsname%%
      \put(2130,436){\makebox(0,0){\strut{}\footnotesize 0.01}}%
      \csname LTb\endcsname%%
      \put(2712,436){\makebox(0,0){\strut{}\footnotesize 0.025}}%
      \csname LTb\endcsname%%
      \put(3680,436){\makebox(0,0){\strut{}\footnotesize 0.05}}%
    }%
    \gplgaddtomacro\gplfronttext{%
      \csname LTb\endcsname%%
      \put(2607,2340){\makebox(0,0)[r]{\strut{}\footnotesize Kaiser}}%
      \csname LTb\endcsname%%
      \put(2607,2100){\makebox(0,0)[r]{\strut{}\footnotesize prolate}}%
      \csname LTb\endcsname%%
      \put(2711,76){\makebox(0,0){\strut{}\small $\tau$}}%
      \csname LTb\endcsname%%
      \put(2711,3063){\makebox(0,0){\strut{}\small $\norm[2]{\bal}$}}%
    }%
    \gplgaddtomacro\gplbacktext{%
      \csname LTb\endcsname%%
      \put(5323,676){\makebox(0,0)[r]{\strut{}\footnotesize $10^{-5}$}}%
      \csname LTb\endcsname%%
      \put(5323,901){\makebox(0,0)[r]{\strut{}\footnotesize $10^{-4}$}}%
      \csname LTb\endcsname%%
      \put(5323,1126){\makebox(0,0)[r]{\strut{}\footnotesize $10^{-3}$}}%
      \csname LTb\endcsname%%
      \put(5323,1352){\makebox(0,0)[r]{\strut{}\footnotesize $10^{-2}$}}%
      \csname LTb\endcsname%%
      \put(5323,1577){\makebox(0,0)[r]{\strut{}\footnotesize $10^{-1}$}}%
      \csname LTb\endcsname%%
      \put(5323,1802){\makebox(0,0)[r]{\strut{}\footnotesize $10^{0}$}}%
      \csname LTb\endcsname%%
      \put(5323,2027){\makebox(0,0)[r]{\strut{}\footnotesize $10^{1}$}}%
      \csname LTb\endcsname%%
      \put(5323,2253){\makebox(0,0)[r]{\strut{}\footnotesize $10^{2}$}}%
      \csname LTb\endcsname%%
      \put(5323,2478){\makebox(0,0)[r]{\strut{}\footnotesize $10^{3}$}}%
      \csname LTb\endcsname%%
      \put(5323,2703){\makebox(0,0)[r]{\strut{}\footnotesize $10^{4}$}}%
      \csname LTb\endcsname%%
      \put(5424,436){\makebox(0,0){\strut{}\footnotesize 0}}%
      \csname LTb\endcsname%%
      \put(5966,436){\makebox(0,0){\strut{}\footnotesize 10}}%
      \csname LTb\endcsname%%
      \put(6508,436){\makebox(0,0){\strut{}\footnotesize 20}}%
      \csname LTb\endcsname%%
      \put(7051,436){\makebox(0,0){\strut{}\footnotesize 30}}%
      \csname LTb\endcsname%%
      \put(7593,436){\makebox(0,0){\strut{}\footnotesize 40}}%
      \csname LTb\endcsname%%
      \put(8135,436){\makebox(0,0){\strut{}\footnotesize 50}}%
    }%
    \gplgaddtomacro\gplfronttext{%
      \csname LTb\endcsname%%
      \put(8710,2524){\makebox(0,0)[r]{\strut{}\footnotesize $0.0$}}%
      \csname LTb\endcsname%%
      \put(8710,2284){\makebox(0,0)[r]{\strut{}\footnotesize $0.01$}}%
      \csname LTb\endcsname%%
      \put(8710,2044){\makebox(0,0)[r]{\strut{}\footnotesize $0.025$}}%
      \csname LTb\endcsname%%
      \put(8710,1805){\makebox(0,0)[r]{\strut{}\footnotesize $0.05$}}%
      \csname LTb\endcsname%%
      \put(8710,1565){\makebox(0,0)[r]{\strut{}\footnotesize $S(\omg)$}}%
      \csname LTb\endcsname%%
      \put(6779,76){\makebox(0,0){\strut{}\small $\omg$}}%
      \csname LTb\endcsname%%
      \put(6779,3063){\makebox(0,0){\strut{}\small $\hat{S}(\omg)$}}%
    }%
    \gplbacktext
    \put(0,0){\includegraphics[width={453.00bp},height={170.00bp}]{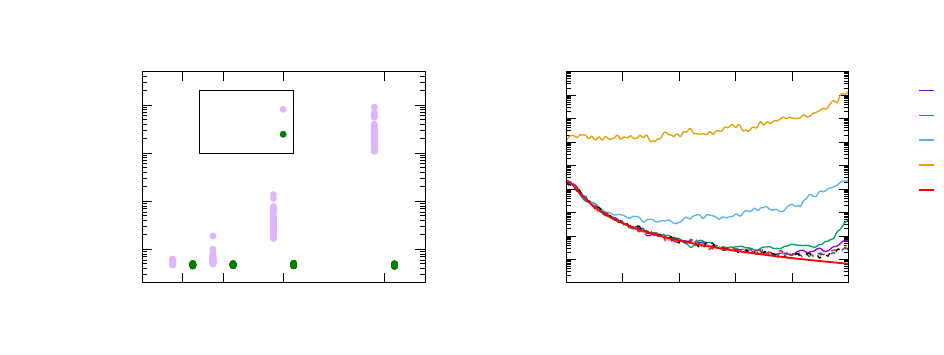}}%
    \gplfronttext
  \end{picture}%
\endgroup
  \caption{Expected values for spectral density estimators computed from
  $n=1,000$ samples of a Mat\'ern process observed $\set{x_j}_{j=1}^n \iid
  \text{Unif}([-1, 1] \setminus [-\tau, 0))$. Left panel: the norm of weights
  $\bal$ with growing $\tau$ (x-axis) for $g$ chosen to be a Kaiser on $[-1, 1]$
  or a prolate on $[-1, 1] \setminus [-\tau, 0)$. Right panel: the expected
  $\hat{S}$ in each case, with colors representing values of $\tau$, solid lines
  giving the estimator using the Kaiser-computed $\bal$, and dotted lines using
  the prolate-computed $\bal$. The true spectral density $S$ is shown in red.
  All weights are computed to resolve up to $\Omega = 50$.}
  \label{fig:norm_est}
\end{figure}

\subsection*{Window reconstuction error}

Another way to understand which window functions $g$ can be well-approximated
using weights $\bal$ with small norm is to more closely inspect the implications
of Theorem~\ref{thm:recovery}. Figure~\ref{fig:thm1} shows the $L^\infty$ errors
in window recovery of a Kaiser function for sampling schemes with and without a
gap. From the left panel, we note that that in the case where the sampling
regime has a gap we fail to recover an $H_{\bal}$ which approximates $G$ well.
However, the right panel illustrates that this failure is due to the limitations
of finite-precision arithmetic and not the mathematical properties of $G$; as
$\norm[2]{\bal} \approx 10^{11}$ for the gapped sampling scheme, the resulting
residuals may be as large as $\norm[\infty]{\Frt \bal - \bm{b}} = 10^{11} \cdot
\epsilon_{\text{mach}} \approx 10^{-5}$ due to catastrophic cancellation when
computing $\Frt \bal$. Nevertheless, the right panel confirms the claims of
Theorem~\ref{thm:recovery} that the residual is small relative to
$\norm[1]{\bal}$ for sufficiently large $n$. This further highlights the
relevance of carefully selecting a window function $g$ such that good
\emph{absolute} error is achieved, which is closely connected to the weights
$\bal$ having a favorably small norm.

\begin{figure}[!t]
  \centering
% GNUPLOT: LaTeX picture with Postscript
\begingroup
  \makeatletter
  \providecommand\color[2][]{%
    \GenericError{(gnuplot) \space\space\space\@spaces}{%
      Package color not loaded in conjunction with
      terminal option `colourtext'%
    }{See the gnuplot documentation for explanation.%
    }{Either use 'blacktext' in gnuplot or load the package
      color.sty in LaTeX.}%
    \renewcommand\color[2][]{}%
  }%
  \providecommand\includegraphics[2][]{%
    \GenericError{(gnuplot) \space\space\space\@spaces}{%
      Package graphicx or graphics not loaded%
    }{See the gnuplot documentation for explanation.%
    }{The gnuplot epslatex terminal needs graphicx.sty or graphics.sty.}%
    \renewcommand\includegraphics[2][]{}%
  }%
  \providecommand\rotatebox[2]{#2}%
  \@ifundefined{ifGPcolor}{%
    \newif\ifGPcolor
    \GPcolortrue
  }{}%
  \@ifundefined{ifGPblacktext}{%
    \newif\ifGPblacktext
    \GPblacktexttrue
  }{}%
  % define a \g@addto@macro without @ in the name:
  \let\gplgaddtomacro\g@addto@macro
  % define empty templates for all commands taking text:
  \gdef\gplbacktext{}%
  \gdef\gplfronttext{}%
  \makeatother
  \ifGPblacktext
    % no textcolor at all
    \def\colorrgb#1{}%
    \def\colorgray#1{}%
  \else
    % gray or color?
    \ifGPcolor
      \def\colorrgb#1{\color[rgb]{#1}}%
      \def\colorgray#1{\color[gray]{#1}}%
      \expandafter\def\csname LTw\endcsname{\color{white}}%
      \expandafter\def\csname LTb\endcsname{\color{black}}%
      \expandafter\def\csname LTa\endcsname{\color{black}}%
      \expandafter\def\csname LT0\endcsname{\color[rgb]{1,0,0}}%
      \expandafter\def\csname LT1\endcsname{\color[rgb]{0,1,0}}%
      \expandafter\def\csname LT2\endcsname{\color[rgb]{0,0,1}}%
      \expandafter\def\csname LT3\endcsname{\color[rgb]{1,0,1}}%
      \expandafter\def\csname LT4\endcsname{\color[rgb]{0,1,1}}%
      \expandafter\def\csname LT5\endcsname{\color[rgb]{1,1,0}}%
      \expandafter\def\csname LT6\endcsname{\color[rgb]{0,0,0}}%
      \expandafter\def\csname LT7\endcsname{\color[rgb]{1,0.3,0}}%
      \expandafter\def\csname LT8\endcsname{\color[rgb]{0.5,0.5,0.5}}%
    \else
      % gray
      \def\colorrgb#1{\color{black}}%
      \def\colorgray#1{\color[gray]{#1}}%
      \expandafter\def\csname LTw\endcsname{\color{white}}%
      \expandafter\def\csname LTb\endcsname{\color{black}}%
      \expandafter\def\csname LTa\endcsname{\color{black}}%
      \expandafter\def\csname LT0\endcsname{\color{black}}%
      \expandafter\def\csname LT1\endcsname{\color{black}}%
      \expandafter\def\csname LT2\endcsname{\color{black}}%
      \expandafter\def\csname LT3\endcsname{\color{black}}%
      \expandafter\def\csname LT4\endcsname{\color{black}}%
      \expandafter\def\csname LT5\endcsname{\color{black}}%
      \expandafter\def\csname LT6\endcsname{\color{black}}%
      \expandafter\def\csname LT7\endcsname{\color{black}}%
      \expandafter\def\csname LT8\endcsname{\color{black}}%
    \fi
  \fi
    \setlength{\unitlength}{0.0500bp}%
    \ifx\gptboxheight\undefined%
      \newlength{\gptboxheight}%
      \newlength{\gptboxwidth}%
      \newsavebox{\gptboxtext}%
    \fi%
    \setlength{\fboxrule}{0.5pt}%
    \setlength{\fboxsep}{1pt}%
    \definecolor{tbcol}{rgb}{1,1,1}%
\begin{picture}(9060.00,3400.00)%
    \gplgaddtomacro\gplbacktext{%
      \csname LTb\endcsname%%
      \put(1255,676){\makebox(0,0)[r]{\strut{}\footnotesize $10^{-15}$}}%
      \csname LTb\endcsname%%
      \put(1255,1014){\makebox(0,0)[r]{\strut{}\footnotesize $10^{-10}$}}%
      \csname LTb\endcsname%%
      \put(1255,1352){\makebox(0,0)[r]{\strut{}\footnotesize $10^{-5}$}}%
      \csname LTb\endcsname%%
      \put(1255,1690){\makebox(0,0)[r]{\strut{}\footnotesize $10^{0}$}}%
      \csname LTb\endcsname%%
      \put(1255,2027){\makebox(0,0)[r]{\strut{}\footnotesize $10^{5}$}}%
      \csname LTb\endcsname%%
      \put(1255,2365){\makebox(0,0)[r]{\strut{}\footnotesize $10^{10}$}}%
      \csname LTb\endcsname%%
      \put(1255,2703){\makebox(0,0)[r]{\strut{}\footnotesize $10^{15}$}}%
      \csname LTb\endcsname%%
      \put(1818,436){\makebox(0,0){\strut{}\footnotesize 100}}%
      \csname LTb\endcsname%%
      \put(2291,436){\makebox(0,0){\strut{}\footnotesize 200}}%
      \csname LTb\endcsname%%
      \put(2763,436){\makebox(0,0){\strut{}\footnotesize 300}}%
      \csname LTb\endcsname%%
      \put(3236,436){\makebox(0,0){\strut{}\footnotesize 400}}%
      \csname LTb\endcsname%%
      \put(3708,436){\makebox(0,0){\strut{}\footnotesize 500}}%
      \csname LTb\endcsname%%
      \put(4180,436){\makebox(0,0){\strut{}\footnotesize 600}}%
    }%
    \gplgaddtomacro\gplfronttext{%
      \csname LTb\endcsname%%
      \put(2768,76){\makebox(0,0){\strut{}\small $n$}}%
      \csname LTb\endcsname%%
      \put(2768,3063){\makebox(0,0){\strut{}\small $\norm[{L^{\infty}([-\Omega, \Omega])}]{H_{\alpha} - G}$}}%
    }%
    \gplgaddtomacro\gplbacktext{%
      \csname LTb\endcsname%%
      \put(5210,676){\makebox(0,0)[r]{\strut{}\footnotesize $10^{-16}$}}%
      \csname LTb\endcsname%%
      \put(5210,929){\makebox(0,0)[r]{\strut{}\footnotesize $10^{-14}$}}%
      \csname LTb\endcsname%%
      \put(5210,1183){\makebox(0,0)[r]{\strut{}\footnotesize $10^{-12}$}}%
      \csname LTb\endcsname%%
      \put(5210,1436){\makebox(0,0)[r]{\strut{}\footnotesize $10^{-10}$}}%
      \csname LTb\endcsname%%
      \put(5210,1690){\makebox(0,0)[r]{\strut{}\footnotesize $10^{-8}$}}%
      \csname LTb\endcsname%%
      \put(5210,1943){\makebox(0,0)[r]{\strut{}\footnotesize $10^{-6}$}}%
      \csname LTb\endcsname%%
      \put(5210,2196){\makebox(0,0)[r]{\strut{}\footnotesize $10^{-4}$}}%
      \csname LTb\endcsname%%
      \put(5210,2450){\makebox(0,0)[r]{\strut{}\footnotesize $10^{-2}$}}%
      \csname LTb\endcsname%%
      \put(5210,2703){\makebox(0,0)[r]{\strut{}\footnotesize $10^{0}$}}%
      \csname LTb\endcsname%%
      \put(5773,436){\makebox(0,0){\strut{}\footnotesize 100}}%
      \csname LTb\endcsname%%
      \put(6246,436){\makebox(0,0){\strut{}\footnotesize 200}}%
      \csname LTb\endcsname%%
      \put(6718,436){\makebox(0,0){\strut{}\footnotesize 300}}%
      \csname LTb\endcsname%%
      \put(7191,436){\makebox(0,0){\strut{}\footnotesize 400}}%
      \csname LTb\endcsname%%
      \put(7663,436){\makebox(0,0){\strut{}\footnotesize 500}}%
      \csname LTb\endcsname%%
      \put(8135,436){\makebox(0,0){\strut{}\footnotesize 600}}%
    }%
    \gplgaddtomacro\gplfronttext{%
      \csname LTb\endcsname%%
      \put(6217,1131){\makebox(0,0)[r]{\strut{}\footnotesize no gap}}%
      \csname LTb\endcsname%%
      \put(6217,891){\makebox(0,0)[r]{\strut{}\footnotesize gap}}%
      \csname LTb\endcsname%%
      \put(6723,76){\makebox(0,0){\strut{}\small $n$}}%
      \csname LTb\endcsname%%
      \put(6723,3063){\makebox(0,0){\strut{}\small $\frac{\norm[{L^{\infty}([-\Omega, \Omega])}]{H_{\alpha} - G}}{\norm[{L^1([-1,1])}]{g} + \norm[1]{\bal}}$}}%
    }%
    \gplbacktext
    \put(0,0){\includegraphics[width={453.00bp},height={170.00bp}]{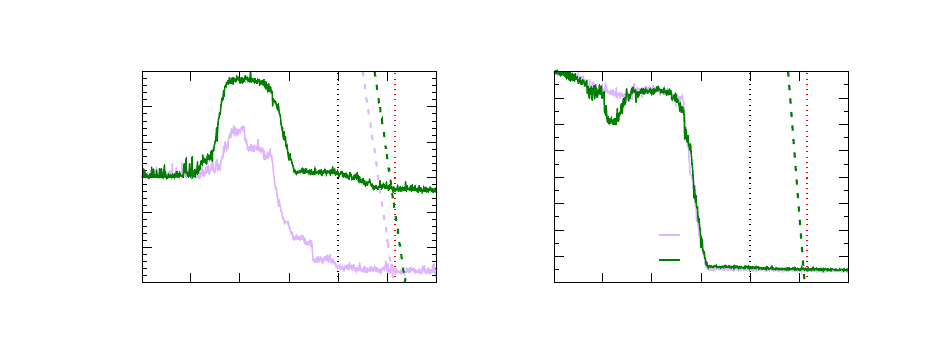}}%
    \gplfronttext
  \end{picture}%
\endgroup
  \caption{Maximum absolute error in Kaiser window recovery with $n$ points
  drawn uniformly at random on $[-1, \tau] \cup [0, 1]$ for $\tau = 0$ (purple)
  and $\tau = 0.15$ (green) for a fixed $\Omega = 400$. The corresponding dotted
  lines show the exponentially decaying error bounds of
  Theorem~\ref{thm:recovery} with $\rho=3$. The dotted black line shows $n =
  4\Omega(b-a)$ which corresponds to Nyquist sampling, and the dotted red line
  shows the bound given in Corollary~\ref{cor:recovery}.}
  \label{fig:thm1}
\end{figure}

\end{document}